\documentclass[journal, onecolumn]{IEEEtran}
\usepackage{epsfig}
\usepackage{times}
\usepackage{float}
\usepackage{afterpage}
\usepackage{amsmath}
\usepackage{amstext}
\usepackage{amssymb,bm}
\usepackage{latexsym}
\usepackage{color}
\usepackage{graphicx}
\usepackage{amsmath}
\usepackage{amsthm}
\usepackage{graphicx}

\usepackage{subfigure}
\usepackage{booktabs}
\usepackage{multicol}
\usepackage{lipsum}
\usepackage{dblfloatfix}
\usepackage{mathrsfs}
\usepackage[textsize=tiny]{todonotes}
\allowdisplaybreaks

\usepackage{cite}
\usepackage{tikz}
\usepackage{pgfplots}
\pgfplotsset{compat=1.12}

\pgfplotsset{compat=1.18}

\newcommand{\rmd}{\mathrm{d}}
\newcommand{\bbE}{\mathbb{E}}\newcommand{\rme}{\mathrm{e}}

\newcommand{\rmi}{\mathrm{i}}

\newcommand{\bbP}{\mathbb{P}}

\newcommand{\bbR}{\mathbb{R}}

\newcommand{\bfA}{\mathbf{A}}
\newcommand{\bfB}{\mathbf{B}}

\newcommand{\bfI}{\mathbf{I}}
\newcommand{\bfJ}{\mathbf{J}}
\newcommand{\bfK}{\mathbf{K}}

\newcommand{\bfm}{\mathbf{m}}

\newcommand{\bfP}{\mathbf{P}}

\newcommand{\bfs}{\mathbf{s}}
\newcommand{\bft}{\mathbf{t}}
\newcommand{\bfU}{\mathbf{U}}\newcommand{\bfu}{\mathbf{u}}
\newcommand{\bfv}{\mathbf{v}}
\newcommand{\bfW}{\mathbf{W}}\newcommand{\bfw}{\mathbf{w}}
\newcommand{\bfX}{\mathbf{X}}\newcommand{\bfx}{\mathbf{x}}\newcommand{\sfX}{\mathsf{X}}
\newcommand{\bfY}{\mathbf{Y}}\newcommand{\bfy}{\mathbf{y}}\newcommand{\sfY}{\mathsf{Y}}
\newcommand{\bfZ}{\mathbf{Z}}

\newcommand{\cI}{\mathcal{I}}

\newcommand{\cN}{\mathcal{N}}

\newcommand{\cX}{\mathcal{X}}
\newcommand{\cY}{\mathcal{Y}}

\newtheoremstyle{mystyle}
{}
{}
{\itshape}
{}
{\bfseries}
{}
{.5em}
{}

\newtheoremstyle{remark}
{}
{}
{}
{}
{\itshape}
{}
{.5em}
{}

\makeatletter
\def\thmhead@plain#1#2#3{%
	\thmname{#1}\thmnumber{\@ifnotempty{#1}{ }\@upn{#2.}}%
	\thmnote{ \small\textsf{\the\thm@notefont\textit{#3}.}}}
\let\thmhead\thmhead@plain
\makeatother

\theoremstyle{mystyle}
\newtheorem{theorem}{Theorem}
\theoremstyle{mystyle}
\newtheorem{lemma}{Lemma}
\theoremstyle{mystyle}
\newtheorem{prop}{Proposition}
\theoremstyle{mystyle}
\newtheorem{corollary}{Corollary}
\theoremstyle{mystyle}
\theoremstyle{remark}
\newtheorem{rem}{Remark}
\theoremstyle{mystyle}
\theoremstyle{mystyle}
\theoremstyle{mystyle}
\theoremstyle{discussion}
\theoremstyle{mystyle}
\theoremstyle{mystyle}

\DeclareMathOperator{\Cov}{Cov}
\DeclareMathOperator{\Tr}{Tr}
\DeclareMathOperator{\mmse}{mmse}
\DeclareMathOperator{\supp}{supp}

\usepackage{enumitem}

\begin{document}

\title{
Fisher Information and Support Complexity of Gaussian Mixtures}

\author{Luca~Barletta,~\IEEEmembership{Member,~IEEE,} and Alex~Dytso,~\IEEEmembership{Senior Member,~IEEE}
\thanks{Luca Barletta is with the Dipartimento di Elettronica, Informazione e Bioingegneria, Politecnico di Milano, 20133 Milano, Italy (e-mail: luca.barletta@polimi.it).}
\thanks{
Alex Dytso is with Qualcomm Flarion Technology, Inc., Bridgewater, NJ 08807, USA (e-mail: odytso2@gmail.com).}
}

\maketitle

\begin{abstract}
We study lower bounds on the Fisher information of Gaussian location mixtures and their applications to estimation problems. First, for a discrete mixing distribution, we derive a universal positive-semidefinite lower bound on the Fisher information matrix in terms of the Rényi entropy of order two. For a mixing distribution supported on $K$ points, this bound scales as $\log(K)/K^2$, improving by a logarithmic factor on the support-size dependence obtained from Stam's inequality. We then develop dimension-aware trace bounds using classical Sobolev inequalities. These bounds are expressed in terms of Sibson's $\alpha$-mutual information and yield the scaling $K^{-2/n}$ in dimension $n$. A complementary bound involving the $\alpha$-mutual information of order $1/2$ is obtained using the sharp Young convolution inequality.

The derived inequalities are further specialized to rotationally invariant mixing distributions supported on discrete spherical shells. In this setting, an intrinsic-dimensional bound on Sibson mutual information leads to lower bounds on the Fisher information in terms of the number of occupied shells. As applications, we study least favorable priors for bounded Gaussian mean estimation. In the scalar case, we show that the least favorable prior on $[-R,R]$ has at least $\Omega(R\sqrt{\log R})$ support points and that each of its probability masses is of order at most $1/R$. In the vector case, we show that the number of occupied shells grows at least linearly with $R$. Finally, Brown's identity and the I-MMSE relationship are used to translate the Fisher-information inequalities into upper bounds on the MMSE and mutual information.
\end{abstract}


\section{Introduction}\label{sec:intro}
Consider the additive Gaussian model
\begin{equation}
Y=X+Z,
\qquad
Z\sim\cN(0,1),
\label{eq:intro_scalar_model}
\end{equation}
where $\cN(0,1)$ denotes the standard normal law and $X$ has a distribution $P_X$ and is
independent of $Z$.  The observation $Y$ is then a Gaussian location mixture. The distribution $P_X$ is often known as \emph{mixing} distribution. 
 One of the goals of this paper  is to study a question of how small the Fisher information of such a mixture can be when the mixing distribution is concentrated on only a few points.

 One way to study this question is by invoking Stam's inequality \cite{stam1959some}. We write $J$
for Fisher information, $h$ for differential entropy, $H$ for the Shannon
entropy of a discrete random variable, and $I$ for mutual information.  Assume that $X$ is a discrete random variable, then Stam's inequality gives us
\begin{align}
J(Y)
&\ge
2\pi\rme\,\exp(-2h(Y))
\\
&=
\exp(-2I(X;Y))
\\
&\ge
\exp(-2H(X)).
\label{eq:intro_stam_entropy}
\end{align}
Hence, small Fisher information requires large Shannon entropy of the mixing
variable.  In particular, if $X$ has at most $K$ mass points, then
$H(X)\le\log K$ and
\begin{equation}
J(Y)\ge K^{-2}.
\label{eq:intro_K_stam}
\end{equation}
Thus a very smooth Gaussian mixture (i.e., a mixture with small Fisher information) cannot be generated by a very small number of components. 

The main theme of this paper is to investigate inequalities of the type in \eqref{eq:intro_stam_entropy} and \eqref{eq:intro_K_stam}. Indeed,  the main scalar result of this paper strengthens this principle.  Let $H_p$
denote Rényi entropy of order $p$.  We show that there exists a universal
constant $c>0$ such that the collision structure of a discrete prior forces
\begin{equation}
J(X+Z)
\ge
c\,H_2(X)\exp(-2H_2(X)).
\label{eq:intro_scalar_new}
\end{equation}  Consequently, a $K$-point mixture satisfies, for another universal
constant $c'>0$,
\begin{equation}
J(X+Z)
\ge
c'\frac{\log K}{K^2}.
\label{eq:intro_scalar_K_new}
\end{equation}
In the scalar case we will also show that 
\begin{equation}
    J(X+Z)
\ge c'' \exp(-2 H_\infty(X) )=c''  \left(\max_x P_{X}(\{x\}) \right)^2
\end{equation}
which, when the support size is fixed,  provides a quantitative indication that the distribution of $X$ is close to uniform.

\subsection{Literature Review}
The connection between entropy and Fisher information goes back at least to Stam's seminal work \cite{stam1959some}. Stam's inequality is known to be equivalent to the Gaussian logarithmic Sobolev inequality; the interested reader is referred to \cite{raginsky2013concentration} and the references therein. Extensions of Stam's inequality that replace differential entropy with R\'enyi entropy can be found in \cite{bobkov2023entropic,wu2025entropic}. In particular, we use the results of \cite{bobkov2023entropic} to derive some of our bounds.  The interested reader is also referred to \cite{lutwak2012extensions} for yet another extension of Stam's inequality.

Inequalities relating Fisher information to the characteristic function play a crucial role in our derivations. In particular, we will show the vector generalization of the scalar inequality established in \cite{zhang2007inequalities}, which lower-bounds Fisher information in terms of the characteristic function. Data processing inequality for the Fisher information  will also play an important role; the interested reader is referred to~\cite{zamir1998proof} for the derivation of this inequality and its application to entropy power inequality. While the present work focuses on nonparametric Fisher information,  inequalities connecting parametric Fisher information to characteristic functions can also be found in \cite{dulek2017characteristic}.

Gaussian mixtures play crucial direct or indirect role in problems studied in information theory and estimation theory. In the context of channel coding, Gaussian mixtures have found applications in the context of point-to-point channels \cite{wu2010impact,alvarado2014high,GaussianBoundsCard,perezcruz2010mimo,Wang_AMplitude}, interference channel \cite{dytso2016interference,CodebookPaper}, and channel modeling \cite{selim2015modeling}  just to name a few. Bounds on entropy of Gaussian mixtures have also received attention and the interested reader is referred to~\cite{eskenazis2024entropy,melbourne2022differential,gMixturesKhandani} and references therein.

Recently, there has also been interest in approximating a Gaussian distribution by Gaussian location mixtures with at most $K$ components, with the approximation error measured using $f$-divergences such as relative entropy and total variation~\cite{BestApporximationGaussianMixture}.  One consequence of our new inequality can be interpreted in the same spirit: it provides a lower bound on the best achievable approximation error when the discrepancy is measured by Fisher divergence.

Gaussian estimation identities will play an important role in our development. Brown-type identities connect the output score and Fisher information to the conditional mean and the minimum mean-square error (MMSE); a modern treatment of these identities for vector Gaussian noise can be found in~\cite{dytso2023conditional}. The I-MMSE identity \cite{guo2005mutual} provides the corresponding connection between mutual information and MMSE.  The interested reader is referred to \cite{perezcruz2010mimo,alvarado2014high,alvarado2018asymptotic} for the high-SNR approximation and bounds of the MMSE of discrete inputs. 

The bounded normal mean problem has a long decision-theoretic history; see \cite{marchand2004review} for a review. Exact solutions in the small-radius regime were characterized in \cite{casella1981estimating}, whereas the large-radius asymptotics and the associated $\cos^2$ profile were derived in \cite{bickel1981minimax}. The same profile also arises naturally in our proof. More recently, an alternative proof that the scalar least favorable prior is unique, symmetric, and finitely discrete was given in \cite{yagli2019estimation}. A multivariate structural theory was developed in \cite{dytso2018structure}; in particular, for a Euclidean ball and Gaussian noise, the least favorable prior was shown to be spherically symmetric and supported on finitely many concentric spheres. Finally, we refer the reader to \cite{olea2026approximate} for computational method for approximating the least favorable prior and minimax estimators.

Some of our development will rely on functional inequalities.  The sharp Euclidean Sobolev constant was obtained independently
by Aubin \cite{aubin1976problemes} and Talenti \cite{talenti1976best}; writing
the inequality for the square root of a density gives a lower bound on Fisher
information in terms of a critical $L^p$ norm, a connection that has also been
used explicitly in information-theoretic form, e.g., in
\cite{yamano2012bounds}.  The geometry-sensitive refinement additionally uses
the sharp Young convolution inequality of Beckner and of Brascamp-Lieb
\cite{beckner1975inequalities,brascamp1976best}.  We specialize the Sobolev
and Young inequalities to discrete Gaussian location mixtures to obtain trace
 bounds.

Sibson's $\alpha$-mutual information will also play an important role in our analysis.  For the introduction of properties of $\alpha$-mutual information, the interested reader is referred to \cite{verdu2015alpha,esposito2025sibson} and references therein.  For the properties of the $\alpha$-mutual information in Gaussian noise, which will be relevant in this work, the interested reader is referred to \cite{milanian2026alpha}.

\subsection{Contributions}
The paper outline and contributions are as follows:
\begin{enumerate}[leftmargin=*]
    \item Section~\ref{sec:preliminaries} collects the preliminary identities and
tools used in our analysis:
\begin{itemize}
\item Section~\ref{sec:estimation_identities} reviews several estimation-theoretic identities, including Brown's identity, which connects the MMSE and Fisher information, and the I-MMSE relationship,
which connects the MMSE and mutual information. 
\item  Section~\ref{sec:char_and_entropy} establishes inequalities relating the
characteristic function of a random variable to its R\'enyi entropy of
order two.
\item  Section~\ref{sec:functional_inequalities} reviews the functional inequalities used in our analysis, including the Sobolev and
Young inequalities.
\end{itemize}
\item Section~\ref{sec:main_results} presents our main results on inequalities for the Fisher information of Gaussian mixtures and shows: 
\begin{itemize}
    \item Section~\ref{sec:bounds_H(X) exp(-2H(X)} establishes a matrix-valued analogue of the scalar inequality~\eqref{eq:intro_scalar_new}. Specifically, we show that the Fisher information matrix is bounded below in the positive semidefinite order by a term of the form $ H_2(\bfX)\exp\bigl(-2H_2(\bfX)\bigr)\mathbf{I}$. 
    \item  Section~\ref{sec:trace_bound} demonstrates that a direct application of the matrix inequality from Section~\ref{sec:bounds_H(X) exp(-2H(X)} to the trace of the Fisher information matrix yields a lower bound with suboptimal dimensional scaling in the exponential term. To overcome this limitation, Section~\ref{sec:trace_bound} establishes a new inequality that captures the correct dependence on the ambient dimension. Specifically, we prove a lower bound on the Fisher information of a Gaussian mixture of the form $ a_n \epsilon \exp\!\left(-\frac{2}{n} I_{\frac{n-2}{n}}(\mathbf{X}; \frac{n-2}{n} \epsilon)\right)$ for $n\ge 3$,  where $I_{\alpha}(\mathbf{X};\gamma)$ denotes the $\alpha$-mutual information of order $\alpha$ with the noise variance $\gamma$. We further derive an alternative bound involving a fixed-order $\alpha$-mutual information, namely, $ a_n \epsilon \exp\!\left(-2 I_{\frac{1}{2}}(\mathbf{X} ; \frac{8\epsilon}{n})\right)$.  This complementary inequality avoids the dimension-dependent order parameter and provides a lower bound in terms of the $\alpha$-mutual information of order $1/2$.
    
\end{itemize}

\item Section~\ref{sec:applications} is dedicated to demonstrating a few applications of the inequalities and shows:
\begin{itemize}
\item  Section~\ref{subsec:bounds_for_shells} applies the new inequalities to mixing distributions supported on spherically symmetric shells, which play an important role in statistics and information theory.
\item  Section~\ref{sec:bounded_mean_problem} applies our bounds to the study of least favorable priors for bounded Gaussian mean estimation. In particular, when the unknown mean is constrained to $[-R,R]$, we show that the least favorable prior must have at least
$ \Omega\left(R\sqrt{\log R}\right) $ mass points. Additionally, it is shown that probability masses are bounded by a term of order $\frac{1}{R}$. 
\item Section~\ref{sec:shell_lower_bound} applies our Fisher-information bounds to the vector bounded Gaussian mean problem, in which the unknown mean is constrained to the Euclidean ball $B_R\subset\bbR^n$ of radius $R$. We show that the least favorable prior must be supported on at least $\Omega(R)$ concentric spheres.
\item Section~\ref{sec:mmse_upper_bounds} shows how the bounds on the Fisher information together with Brown's identity can be used to produce upper bounds on the MMSE. 
\end{itemize}

\end{enumerate}

\subsection{Notation}
Vectors are denoted by boldface symbols.  We use $\bbE$ and $\bbP$ for
expectation and probability, $\Tr$ for the matrix trace, and
$\rmi:=\sqrt{-1}$.  The Euclidean norm is $\|\cdot\|$, $\bfI_n$ is the
$n\times n$ identity matrix, and $\bfA\succeq\bfB$
(resp. $\bfA\preceq\bfB$) means that $\bfA-\bfB$ (resp. $\bfB-\bfA$)
is positive semidefinite.  We write $\cN(\bfm,\bfK)$ for a Gaussian law
with mean $\bfm$ and covariance $\bfK$.  The notation $\delta_{\bfx}$
denotes the unit point mass at $\bfx$, and $\mathbf 1_A$ denotes the
indicator of an event or set $A$.  Vertical bars denote cardinality for finite
sets and Lebesgue or surface measure for geometric sets, as dictated by the
context.  Finally,
\begin{equation}
B_R:=\{\bfx\in\bbR^n:\|\bfx\|\le R\}.
\end{equation}
For $m\ge0$, $\mathbb S^m$ denotes the unit sphere in $\bbR^{m+1}$.
Writing $\Gamma$ for Euler's gamma function,
\begin{equation}
\omega_n:=|B_1|=\frac{\pi^{n/2}}{\Gamma(n/2+1)}.
\end{equation}
For a discrete random vector $\bfX$ with masses $\{p_j\}$, its R\'enyi entropy
of order $p\in(0,\infty)\setminus\{1\}$ is
\begin{equation}
H_p(\bfX)
:=\frac{1}{1-p}\log\sum_jp_j^p,
\end{equation}
with the standard continuous extensions at $p=0,1,\infty$, that is
\begin{align}
H_0(\bfX)
&=
\log\left|\supp(P_{\bfX})\right| \\
H_1(\bfX) &= - \sum_j p_j \log p_j\\
H_\infty(\bfX)
&= -\log\max_j p_j.
\end{align}  We write $H(\bfX):=H_1(\bfX)$ for
Shannon entropy.  For an absolutely continuous random vector $\bfW$ with
density $f_{\bfW}$, $h(\bfW)$ denotes differential entropy and, for
$p\in(0,\infty)\setminus\{1\}$,
\begin{equation}
h_p(\bfW)
:=
\begin{cases}
\displaystyle
\frac{1}{1-p}
\log\int_{\bbR^n}f_{\bfW}(\bfw)^p\,\rmd\bfw,
& p\in(0,\infty)\setminus\{1\},
\\[2ex]
\displaystyle
-\log\left\|f_{\bfW}\right\|_\infty,
& p=\infty.
\end{cases}
\end{equation}
denotes differential R\'enyi entropy.  We use $I(\bfX;\bfY)$ for mutual
information.

If a random vector $\bfW$ has a positive differentiable density $f_{\bfW}$,
its score and Fisher information matrix are
\begin{equation}
\rho_{\bfW}(\bfw)=\nabla\log f_{\bfW}(\bfw),
\quad
\bfJ(\bfW)=\bbE[\rho_{\bfW}(\bfW)\rho_{\bfW}(\bfW)^T],
\end{equation}
and the scalar Fisher information is
\begin{equation}
J(\bfW)=\Tr(\bfJ(\bfW)).
\end{equation}
When no ambiguity is possible, we also write $J(f_{\bfW})$ and
$\bfJ(f_{\bfW})$.  The characteristic function is
\begin{equation}
\varphi_{\bfW}(\bft)
:=
\bbE\left[\rme^{\rmi\bft^T\bfW}\right].
\end{equation}
For $t>0$, we denote the centered Gaussian density with covariance
$t^2\bfI_n$ by
\begin{equation}
\phi_t(\bfx)
:=
\frac{1}{(2\pi t^2)^{n/2}}
\exp\left(-\frac{\|\bfx\|^2}{2t^2}\right).
\end{equation}
For a measurable function $g$, $\|g\|_p$ denotes its $L^p(\bbR^n)$ norm,
and $\langle f,g\rangle:=\int_{\bbR^n}f(\bfx)g(\bfx)\rmd\bfx$
whenever the integral is well defined.  The Gaussian $Q$-function is
\begin{equation}
Q(x):=\frac{1}{\sqrt{2\pi}}\int_x^\infty \rme^{-u^2/2}\rmd u.
\end{equation}
All logarithms are natural.

\section{Preliminary Results and Inequalities}\label{sec:preliminaries}
In this section, we collect preliminary inequalities and identities needed in our analysis.  We begin by reviewing some estimation theoretic identities that connect the MMSE and the Fisher information.

\subsection{Gaussian Estimation Identities}
\label{sec:estimation_identities}
Consider a random vector $\bfX$ with finite second moment and, for $\gamma>0$,
\begin{equation}
\bfY_\gamma=\sqrt{\gamma}\bfX+\bfZ
\label{eq:prelim_gaussian_channel}
\end{equation}
with $\bfZ\sim\cN(0,\bfI_n)$ independent of $\bfX$. The MMSE matrix is defined as 
\begin{align}
{\bf MMSE}(\bfX|\bfY_\gamma)
:=\bbE\big[&(\bfX-\bbE[\bfX|\bfY_\gamma])
(\bfX-\bbE[\bfX|\bfY_\gamma])^T\big]
\end{align}
and the MMSE is defined as 
\begin{equation}
\mmse(\bfX|\bfY_\gamma)
:=\Tr({\bf MMSE}(\bfX|\bfY_\gamma)).
\end{equation}
The MMSE and the Fisher information matrix of Gaussian mixture $\bfY_\gamma$ are connected through
Brown's identity \cite{dytso2023conditional}: for $\gamma>0$
\begin{equation}
\gamma {\bf MMSE}(\bfX|\bfY_\gamma)
=
\bfI_n-\bfJ(\bfY_\gamma),
\label{eq:brown_identity}
\end{equation}
  We shall also use the I-MMSE
identity \cite{guo2005mutual}: for $\gamma>0$
\begin{equation}
\frac{\rmd}{\rmd\gamma}I(\bfX;\bfY_\gamma)
=
\frac12\mmse(\bfX|\bfY_\gamma).
\label{eq:I_MMSE_prelim}
\end{equation}

Finally, we will need the following inequality for the Fisher information of projections, which follows from the standard data-processing inequality for Fisher information \cite{zamir1998proof}. 

\begin{lemma}[Fisher information under orthogonal projections]
\label{lem:projection_fisher_DPI}
Let $\bfW\in\bbR^n$ have a density with
finite Fisher information, and let $\bfP\in\bbR^{m\times n}$ satisfy
\begin{equation}
\bfP\bfP^T=\bfI_m.
\end{equation}
Then
\begin{equation}
\bfJ(\bfP\bfW)
\preceq
\bfP\bfJ(\bfW)\bfP^T.
\label{eq:projection_matrix_fisher_DPI}
\end{equation}
Consequently,
\begin{equation}
J(\bfP\bfW)
\le
\Tr\big(\bfP\bfJ(\bfW)\bfP^T\big).
\label{eq:projection_fisher_DPI}
\end{equation}
\end{lemma}

We use \eqref{eq:projection_fisher_DPI} repeatedly to turn lower-dimensional
trace inequalities into spectral statements.

\subsection{Fisher Information and Characteristic Functions}
\label{sec:char_and_entropy}

We now provide the following matrix generalization of the inequality shown in \cite{zhang2007inequalities}; see also \cite{bobkov2014fisher,tepedelenlioglu2010inequalities} for  alternative proofs. This inequality will be a key bridge between  characteristic function and Fisher information.

\begin{prop}\label{prop:matrix_cf_fisher}
Let $\bfW\in\bbR^n$ have an absolutely continuous, differentiable density
for which the integrations by parts below are valid, and assume that
$\bfJ(\bfW)$ is finite.  Then, for all $\bft\in\bbR^n$,
\begin{equation}
|\varphi_{\bfW}(\bft)|^2\bft\bft^T
\preceq
\big(1-|\varphi_{\bfW}(\bft)|^2\big)\bfJ(\bfW).
\label{eq:matrix_cf_fisher}
\end{equation}
\end{prop}

\begin{IEEEproof}
Integration by parts gives
\begin{align}
\bbE\big[\rho_{\bfW}(\bfW) \rme^{\rmi\bft^T\bfW}\big]
&=
\int_{\bbR^n}\nabla f_{\bfW}(\bfw)\rme^{\rmi\bft^T\bfw}\rmd\bfw
\\
&=-\rmi\bft\,\varphi_{\bfW}(\bft).
\label{eq:score_cf_identity}
\end{align}
Moreover $\bbE[\rho_{\bfW}(\bfW)]=0$.  Hence
\begin{equation}
\bbE\left[\rho_{\bfW}(\bfW)
\big(\rme^{\rmi\bft^T\bfW}-\varphi_{\bfW}(\bft)\big)\right]
=-\rmi\bft\,\varphi_{\bfW}(\bft).
\end{equation}
For arbitrary $\bfu\in\bbR^n$, Cauchy-Schwarz yields
\begin{align}
|\bfu^T\bft|^2|\varphi_{\bfW}(\bft)|^2
&\le
\bbE \left[|\bfu^T\rho_{\bfW}(\bfW)|^2 \right]
\bbE \left[|\rme^{\rmi\bft^T\bfW}-\varphi_{\bfW}(\bft)|^2 \right]
\\
&=
\bfu^T\bfJ(\bfW)\bfu
\big(1-|\varphi_{\bfW}(\bft)|^2\big).
\end{align}
Since this holds for every $\bfu$, \eqref{eq:matrix_cf_fisher} follows.
\end{IEEEproof}

For later use, consider the one-dimensional model $Y=X+Z$ with
$Z\sim\cN(0,1)$ independent of $X$. The scalar
specialization of Proposition~\ref{prop:matrix_cf_fisher} gives
$|\varphi_Y(t)|^2t^2\le J(Y)(1-|\varphi_Y(t)|^2)$, and hence
\begin{equation}
|\varphi_X(t)|^2
\le
\frac{J(Y)\rme^{t^2}}{J(Y)+t^2},
\qquad t\in\bbR,
\label{eq:scalar_cf_fisher_prelim}
\end{equation}
where we used $\varphi_Y(t)=\rme^{-t^2/2}\varphi_X(t)$.

\subsection{R\'enyi Entropy and Characteristic Function}
In this subsection, we present a result that will connect R\'enyi entropy of $\bfX$ to the characteristic function of $\bfX$.  We start by defining the following function:  for $T>0$, let
\begin{equation}
F_T(\bft)
:=
\frac{|B_{\frac{T}{2}}\cap(B_{\frac{T}{2}}+\bft)|}{|B_{\frac{T}{2}}|}, \, \bft \in \bbR^n.
\label{eq:def_fejer_ball}
\end{equation}
The function $F_T$ is supported on $B_T$, takes values in $[0,1]$, and has a
nonnegative Fourier transform. Note that in the scalar case $F_T$ becomes the triangular function 
\begin{equation}
    F_T(t) = \left(1 - \frac{|t|}{T} \right)_{+}, \, t \in \bbR,
\end{equation}
where $(x)_+  =\max(0,x)$.   Additionally, note that $F_T(\bft)$ admits the following representation as a
normalized autocorrelation:
\begin{equation}
F_T(\bft)
=
\frac{1}{|B_{T/2}|}
\int_{\bbR^n}
\mathbf 1_{B_{T/2}}(\bfs)
\mathbf 1_{B_{T/2}}(\bfs-\bft)\,\rmd\bfs.
\label{eq:fejer_autocorrelation}
\end{equation}
Consequently, by the convolution theorem, its Fourier transform is given by: for $ \bfu\in\bbR^n$,
\begin{equation}
\int_{\bbR^n}
F_T(\bft) \rme^{\rmi\bft^T\bfu}\,\rmd\bft
=
\frac{1}{|B_{T/2}|}
\left|
\int_{B_{T/2}}
\rme^{\rmi\bfs^T\bfu}\,\rmd\bfs
\right|^2
\ge 0,
\label{eq:fejer_positive_transform}
\end{equation}
Thus, $F_T$ has a nonnegative Fourier transform.

\begin{prop}
\label{prop:fejer_discrete}
Let $\bfX \in \bbR^n$ be a discrete random vector.  Then, for every $T>0$,
\begin{equation}
\int_{B_T}F_T(\bft)|\varphi_{\bfX}(\bft)|^2\rmd\bft
\ge
\omega_n\left(\frac{T}{2}\right)^n \rme^{-H_2(\bfX)}.
\label{eq:fejer_discrete}
\end{equation}
\end{prop}

\begin{IEEEproof} Let 
\begin{equation}
P_{\bfX}=\sum_{j\ge1}p_j\delta_{\bfx_j}
\end{equation}
denote the probability mass function of $\bfX$.
Now, expanding the squared characteristic function, we get 
\begin{align}
&\int_{B_T}F_T(\bft)|\varphi_{\bfX}(\bft)|^2\rmd\bft
=
\sum_{i,j}p_ip_j
\int_{B_T}F_T(\bft) \rme^{\rmi\bft^T(\bfx_i-\bfx_j)}\rmd\bft.
\label{eq:fejer_expand}
\end{align}
The interchange is justified by Fubini since $F_T$ is integrable.  Moreover, from \eqref{eq:fejer_positive_transform} 
for every $\bfu\in\bbR^n$,
\begin{align}
\int_{\bbR^n}F_T(\bft) \rme^{\rmi\bft^T\bfu}\rmd\bft \ge0,
\end{align}
which implies that every summand in \eqref{eq:fejer_expand} is nonnegative.  Consequently, retaining
only $i=j$ gives
\begin{align}
\int_{B_T}F_T(\bft)|\varphi_{\bfX}(\bft)|^2\rmd\bft
&\ge
\sum_jp_j^2\int_{B_T}F_T(\bft)\rmd\bft
\\
&=
|B_{\frac{T}{2}}|\sum_jp_j^2,
\end{align}
which is \eqref{eq:fejer_discrete}. This concludes the proof. 
\end{IEEEproof}

\subsection{$\alpha$-mutual information}
The Sibson $\alpha$-mutual information \cite{sibson1969information,verdu2015alpha} between two random variables $(\sfX,\sfY) \in \cX \times \cY$ with  joint distribution $P_{\sfX,\sfY}$ is defined as follows: for $\alpha  \in [0,\infty]$
\begin{equation}
I_{\alpha}(\sfX;\sfY) 
 = \inf_{Q_{\sfY}} D_{\alpha}\left( P_{\sfX,\sfY} \| P_{\sfX} \times Q_{\sfY} \right ), \label{eq:Def_I}
\end{equation}
where  $D_{\alpha}(\cdot \| \cdot)$ is R\'enyi  divergence of order $\alpha$, which for $P \ll Q$ reads
\begin{equation}
D_{\alpha}\left(P \| Q \right ) = \frac{1}{\alpha-1}  \log \int p^\alpha q^{1-\alpha} \ {\rmd}\mu, 
\end{equation}
where $p$ and $q$ are the Radon-Nikodym derivatives of the probability measures $P$ and $Q$, respectively, with respect to a common dominating $\sigma$-finite measure $\mu$.  

For the case of Gaussian noise we use the following notation: for $\epsilon >0$
\begin{equation}
    I_\alpha( \bfX ; \epsilon) = I_\alpha \left(\bfX; \bfX +\sqrt{  \frac{1}{\epsilon} } \bfZ \right) 
\end{equation}
for $\alpha=1$,  we use $I(\bfX;\epsilon)=I_1(\bfX;\epsilon)$ to denote the classical mutual information. 

We will also need $\alpha$-mutual information of order zero. However, it is not difficult to check that
\begin{equation}
    \lim_{\alpha \to 0} I_{\alpha}(\bfX; \epsilon)=0.
\end{equation}
A more useful way to define $\alpha$-mutual information of order zero is as follows
\begin{equation}
    I_{0}(\bfX, \epsilon) = \lim_{\alpha \to 0} I_{\alpha} \left(\bfX; \frac{\epsilon}{\alpha} \right).
\end{equation} 

The next proposition summarizes some of the important properties of $I_\alpha( \bfX ; \epsilon)$.

\begin{prop}[Properties of Gaussian Sibson mutual information]
\label{prop:sibson_properties}
Let $\bfX\in\bbR^n$, let $\bfZ\sim\cN(0,\bfI_n)$ be independent of
$\bfX$, and let $\epsilon>0$. Then the following properties hold.
\begin{enumerate}
 \item $\alpha\mapsto I_\alpha(\bfX;\epsilon)$ is nondecreasing on $(0,\infty)$.

    \item For every $\alpha>0$,
    \begin{equation}
    I_\alpha(\bfX;\epsilon)
    =
    h_{\frac{1}{\alpha}}\left(
    \bfX+\frac{1}{\sqrt{\alpha\epsilon}}\bfZ
    \right)
    -
    h_{\frac{1}{\alpha}}\left(
    \frac{1}{\sqrt{\alpha\epsilon}}\bfZ
    \right),
    \label{eq:sibson_renyi_entropy_representation}
    \end{equation}

    \item If $\bfX$ is discrete, then, for every $\alpha>0$,
    \begin{equation}
    I_\alpha(\bfX;\epsilon)
    \le
    H_{\frac{1}{\alpha}}(\bfX),
    \label{eq:sibson_discrete_entropy_bound}
    \end{equation}
    and
    \begin{equation}
    \lim_{\epsilon\to\infty}I_\alpha(\bfX;\epsilon)
    =
    H_{\frac{1}{\alpha}}(\bfX).
    \label{eq:sibson_high_snr_limit}
    \end{equation}
\item  For $\alpha=0$, we have that 
  \begin{equation}
I_0(\bfX;\epsilon)
=
h_\infty\left(
\bfX+\frac{1}{\sqrt{\epsilon}}\bfZ
\right)
-
h_\infty\left(
\frac{1}{\sqrt{\epsilon}}\bfZ
\right).
\label{eq:order_zero_mutual_information_renyi}
\end{equation}

\item If $\bfX'$ is an independent copy of $\bfX$, then
    \begin{equation}
    I_{1/2}(\bfX;\epsilon)
    =
    -\log
    \bbE\left[
    \exp\left(
    -\frac{\epsilon}{8}\|\bfX-\bfX'\|^2
    \right)
    \right].
    \label{eq:sibson_half_gaussian_kernel}
    \end{equation}

\end{enumerate}
\end{prop}
\begin{IEEEproof}
The proof of the monotonicity in $\alpha$ and \eqref{eq:sibson_discrete_entropy_bound} can be found in \cite{verdu2015alpha}. The proof of \eqref{eq:sibson_renyi_entropy_representation} and \eqref{eq:sibson_high_snr_limit} can be found in \cite{milanian2026alpha}.

For $\alpha>0$, from \eqref{eq:sibson_renyi_entropy_representation} we have that 
\begin{equation}
I_\alpha(\bfX;\eta)
=
h_{\frac{1}{\alpha}}\left(
\bfX+\frac{1}{\sqrt{\alpha\eta}}\bfZ
\right)
-
h_{\frac{1}{\alpha}}\left(
\frac{1}{\sqrt{\alpha\eta}}\bfZ
\right).
\end{equation}
Setting $\eta=\epsilon/\alpha$, so that $\alpha\eta=\epsilon$, yields
\begin{equation}
I_\alpha\left(\bfX;\frac{\epsilon}{\alpha}\right)
=
h_{\frac{1}{\alpha}}\left(
\bfX+\frac{1}{\sqrt{\epsilon}}\bfZ
\right)
-
h_{\frac{1}{\alpha}}\left(
\frac{1}{\sqrt{\epsilon}}\bfZ
\right).
\end{equation}
Therefore, letting
$\alpha\downarrow0$, so that $1/\alpha\to\infty$, and using the definition
of $I_0(\bfX;\epsilon)$, we obtain
\begin{equation}
I_0(\bfX;\epsilon)
=
h_\infty\left(
\bfX+\frac{1}{\sqrt{\epsilon}}\bfZ
\right)
-
h_\infty\left(
\frac{1}{\sqrt{\epsilon}}\bfZ
\right).
\end{equation}

The proof of \eqref{eq:sibson_half_gaussian_kernel} can be found in  Appendix~\ref{app:proof_prop}.
\end{IEEEproof}

\begin{lemma}[Chain-rule-type  inequality]
\label{lem:sibson_mixture}
Let $V$ be a discrete random variable, and let $X$ and $Y$ be
arbitrary random variables. Then, for every $\alpha>0$,
\begin{equation}
I_\alpha(V,X;Y)
\le
H_{\frac{1}{\alpha}}(V)
+
\sup_{v\in\supp(P_V)}
I_\alpha(X;Y\mid V=v),
\label{eq:sibson_mixture}
\end{equation}
where 
\begin{equation}
I_\alpha(X;Y\mid V=v)
=
\inf_{Q_Y}
D_\alpha\left(
P_{X,Y\mid V=v}
\,\middle\|\,P_{X\mid V=v}\times Q_Y\right).
\end{equation}

\end{lemma}

\begin{IEEEproof}
For $\alpha\ne1$, the closed-form expression for Sibson mutual
information and the power-sum inequality give
\begin{align}
&I_\alpha(V,X;Y)
\le \nonumber\\
&\frac{\alpha}{\alpha-1}
\log\sum_{v}
P_V(v)^{\frac{1}{\alpha}}
\exp\left(
\frac{\alpha-1}{\alpha}
I_\alpha(X;Y\mid V=v)
\right).
\label{eq:sibson_mixture_refined}
\end{align}
Let
\begin{equation}
M=\sup_{v\in\supp(P_V)}
I_\alpha(X;Y\mid V=v).
\end{equation}
Bounding the conditional information terms in
\eqref{eq:sibson_mixture_refined} by $M$, while accounting for the
sign of $\alpha-1$, yields
\begin{align}
I_\alpha(V,X;Y)
&\le
M+
\frac{\alpha}{\alpha-1}
\log\sum_vP_V(v)^{\frac{1}{\alpha}}\\
&=
M+H_{\frac{1}{\alpha}}(V),
\end{align}
which proves \eqref{eq:sibson_mixture}. The case $\alpha=1$ follows
by continuity, or directly from the usual Shannon chain rule.
\end{IEEEproof}

\subsection{Isoperimetric Inequalities and Classical Functional Inequalities}
\label{sec:functional_inequalities}

We next recall an information-theoretic isoperimetric inequality and
several classical functional inequalities. For an absolutely continuous
random vector $\bfX\in\bbR^n$ with finite differential entropy and
Fisher information, Stam's inequality states that
\begin{equation}
J(\bfX)
\ge
2\pi\rme n
\exp\left(-\frac{2}{n}h(\bfX)\right).
\label{eq:stam_classical}
\end{equation}
Applying \eqref{eq:stam_classical} to
$\bfX+\epsilon^{-1/2}\bfZ$ and using
\begin{equation}
I(\bfX;\epsilon)
=
h\left(\bfX+\frac{1}{\sqrt{\epsilon}}\bfZ\right)
-\frac{n}{2}\log\frac{2\pi\rme}{\epsilon},
\end{equation}
we obtain the equivalent form
\begin{equation}
J\left(\bfX+\frac{1}{\sqrt{\epsilon}}\bfZ\right)
\ge
\epsilon n
\exp\left(-\frac{2}{n}I(\bfX;\epsilon)\right).
\label{eq:stam_vector_application}
\end{equation}

We next record the classical functional inequalities used in our analysis.
For completeness, we recall that
\begin{equation}
	H^1(\bbR^n)
	=
	\left\{
		g\in L^2(\bbR^n):
		\nabla g\in L^2(\bbR^n;\bbR^n)
	\right\},
\end{equation}
where the gradient is understood in the weak sense. All the inequalities
below are stated for functions in this space.
\begin{itemize}
    \item  \emph{Elementary Sobolev inequality:} Let $n=1$. For every $g\in  H^1(\bbR)$, 
    \begin{equation}
    \|g\|_\infty^2
    \le \|g\|_2\|g'\|_2.
    \label{eq:sobolev_1d}
\end{equation}
Indeed, since $g^2$ is absolutely continuous
and vanishes at both infinities,
\begin{align}
    2|g(x)|^2
    &\le \int_{\bbR}|(g^2)'(t)|\,\rmd t\\
    &\le 2\|g\|_2\|g'\|_2.
\end{align}

	\item \emph{Sharp Sobolev inequality \cite{aubin1976problemes,talenti1976best}:}
	Let $n\ge 3$. For every $g\in  H^1(\bbR^n)$,
	\begin{equation}
		\|g\|_{2n/(n-2)}^2
		\le
		\frac{4}{n(n-2)|\mathbb S^n|^{2/n}}
		\|\nabla g\|_2^2.
		\label{eq:sharp_sobolev_prelim}
	\end{equation}

	\item \emph{Ladyzhenskaya-Gagliardo-Nirenberg inequality \cite{ladyzhenskaya1969mathematical}:}
	For every $g\in H^1(\bbR^2)$,
	\begin{equation}
		\|g\|_4^4
		\le
		2\|g\|_2^2\|\nabla g\|_2^2.
		\label{eq:ladyzhenskaya_prelim}
	\end{equation}

	\item \emph{Sharp Young convolution inequality \cite{beckner1975inequalities,brascamp1976best}:}
	Let $1\le p,q,r\le\infty$ satisfy
	\begin{equation}
		\frac{1}{p}+\frac{1}{q}
		=
		1+\frac{1}{r}.
	\end{equation}
	Then, for every $f\in L^p(\bbR^n)$ and
	$g\in L^q(\bbR^n)$,
	\begin{equation}
		\|f*g\|_r
		\le
		\left(\frac{C_pC_q}{C_r}\right)^n
		\|f\|_p\|g\|_q,
		\label{eq:sharp_young_prelim}
	\end{equation}
	where
	\begin{equation}
		C_s^2
		=
		\frac{s^{1/s}}{(s')^{1/s'}},
		\qquad
		\frac{1}{s}+\frac{1}{s'}=1,
	\end{equation}
	and $s'$ denotes the H\"older conjugate of $s$. At the
	endpoints, we use the conventions $1'=\infty$, $\infty'=1$,
	and $C_1=C_\infty=1$. 
\end{itemize}

Applying the Sobolev inequalities and the
Ladyzhenskaya-Gagliardo-Nirenberg inequality to
$g=\sqrt{f_{\bfX}}$, where $f_{\bfX}$ is the density of $\bfX$, yields the following
R\'enyi-entropy generalization of Stam's
inequality. Specifically, we use
\eqref{eq:sobolev_1d} for $n=1$,
\eqref{eq:ladyzhenskaya_prelim} for $n=2$,
and \eqref{eq:sharp_sobolev_prelim} for
$n\ge3$. We refer the interested reader
to \cite{bobkov2023entropic} for further
details and related derivations.

\begin{prop}
\label{prop:fisher_renyi_entropy_power}
Let $n\ge1$, and let $\bfX\in\bbR^n$ have an absolutely continuous
density and finite Fisher information. Define
\begin{equation}
r_n
:=
\begin{cases}
\infty, & n=1,\\
2, & n=2,\\
\displaystyle\frac{n}{n-2}, & n\ge3.
\end{cases} \label{eq:rn_definition}
\end{equation}
For $\alpha\in(1,\infty]$, define the R\'enyi entropy power by
\begin{equation}
N_\alpha(\bfX)
:=
\begin{cases}
\displaystyle
\frac{1}{2\pi}
\alpha^{-1/(\alpha-1)}
\exp\left(
\frac{2}{n}h_\alpha(\bfX)
\right),
& \alpha\in(1,\infty),
\\[2ex]
\displaystyle
\frac{1}{2\pi}
\exp\left(
\frac{2}{n}h_\infty(\bfX)
\right),
& \alpha=\infty.
\end{cases}
\label{eq:def_normalized_renyi_entropy_power}
\end{equation}
Then
\begin{equation}
J(\bfX)
\ge
\frac{a_n}{N_{r_n}(\bfX)},
\label{eq:fisher_renyi_entropy_power}
\end{equation}
where
\begin{equation}
a_n
:=
\begin{cases}
\displaystyle\frac{2}{\pi},
& n=1,
\\[1.5ex]
\displaystyle\frac{1}{2\pi},
& n=2,
\\[2ex]
\displaystyle
\frac{n(n-2)}{2\pi}
|\mathbb S^n|^{2/n}
\left(\frac{n-2}{n}\right)^{(n-2)/2},
& n\ge3.
\end{cases}
\label{eq:fisher_renyi_entropy_power_constant}
\end{equation}
\end{prop}

\section{Inequalities for Fisher Information of Gaussian Mixtures}\label{sec:main_results}
In this section, we present the main lower bounds on the Fisher information of Gaussian mixtures. We begin with a universal matrix bound for discrete mixing distributions. 

\subsection{Lower Bound on the Fisher Information Matrix}
\label{sec:bounds_H(X) exp(-2H(X)}

We begin with the scalar directional bound.
The proof combines Proposition~\ref{prop:fejer_discrete} with the pointwise
characteristic-function constraint \eqref{eq:scalar_cf_fisher_prelim}.  

Define the universal constant
\begin{equation}
c_{\rm F}:=\frac{1}{\pi^2}.
\label{eq:def_cF}
\end{equation}

\begin{theorem}[Universal discrete-mixture Fisher bound]
\label{thm:universal_fisher_bound}
Let $\bfX\in\bbR^n$ be discrete and let
$\bfZ\sim\cN(0,\bfI_n)$ be independent of $\bfX$.  Then, for every $\epsilon>0$,
\begin{align}
\bfJ\left(\bfX+\frac{1}{\sqrt{\epsilon}}\bfZ\right)
&\succeq
\epsilon
\exp(-2H_2(\bfX)) 
\max\left\{\frac{2}{\pi},c_{\rm F}H_2(\bfX)\right\}\bfI_n.
\label{eq:universal_vector_fisher_t}
\end{align}
\end{theorem}
\begin{IEEEproof}
By scaling property of Fisher information, it suffices to consider $\epsilon=1$. We first prove the
scalar assertion. Let $Y=X+Z$ and write
\begin{equation}
H=H_2(X),\qquad \eta=\rme^{-H}=\sum_j p_j^2,
\qquad J=J(Y).
\end{equation}

We begin with a bound valid for every $H\ge0$. The one-dimensional
Sobolev inequality~\eqref{eq:sobolev_1d}, applied to the square root
of the density of $Y$, gives
\begin{equation}
J\ge 4\|f_Y\|_\infty^2.
\end{equation}
Since $f_Y(x_j)\ge p_j/\sqrt{2\pi}$ at every atom $x_j$ of $X$,
\begin{align}
J&\ge\frac2\pi\left(\max_jp_j\right)^2\\
&\ge\frac2\pi\left(\sum_jp_j^2\right)^2
=\frac2\pi\rme^{-2H},
\label{eq:improved_constant_branch}
\end{align}
where the second inequality uses $\sum_jp_j^2\le\max_jp_j$.
In particular, since $6<2\pi$, for $0\le H\le6$ we also have
\begin{equation}
J\ge \frac{H}{\pi^2}\rme^{-2H}.
\label{eq:improved_entropy_smallH}
\end{equation}

It remains to prove the entropy-enhanced bound for $H\ge6$. By
\eqref{eq:scalar_cf_fisher_prelim} and the scalar case of
Proposition~\ref{prop:fejer_discrete}, for every $T>0$,
\begin{equation}
\eta T\le
\int_{-T}^{T}\left(1-\frac{|u|}{T}\right)
\frac{J\rme^{u^2}}{J+u^2}\,\rmd u.
\label{eq:improved_fejer_step}
\end{equation}
Take $T\ge2$ and split the integral over
$|u|\le1$, $1<|u|\le T/2$, and $T/2<|u|\le T$.
For the first part, convexity of the exponential gives
$\rme^{u^2}\le1+(\rme-1)u^2$ when $|u|\le1$. Thus,
\begin{align}
\int_{-1}^{1}\left(1-\frac{|u|}{T}\right)
\frac{J\rme^{u^2}}{J+u^2}\,\rmd u &\le \int_{\bbR}\frac{J}{J+u^2}\,\rmd u
+2(\rme-1)J\nonumber\\
&=\pi\sqrt J+2(\rme-1)J.
\label{eq:improved_integral_first}
\end{align}
On the second part, using $(J+u^2)^{-1}\le u^{-2}$,
\begin{align}
2\int_1^{T/2}\left(1-\frac uT\right)
\frac{J\rme^{u^2}}{J+u^2}\,\rmd u &\le 2J\rme^{T^2/4}
\int_1^{T/2}\frac{\rmd u}{u^2} \nonumber \\
&\le 2J\rme^{T^2/4}.
\label{eq:improved_integral_second}
\end{align}
On the last part, $u^2\ge T^2/4$ and
$T^2-u^2\ge T(T-u)$. With the substitution $v=T-u$,
\begin{align}
2\int_{T/2}^{T}\left(1-\frac uT\right)
\frac{J\rme^{u^2}}{J+u^2}\,\rmd u&\le\frac{8J\rme^{T^2}}{T^3}
\int_0^{T/2}v\rme^{-Tv}\,\rmd v\nonumber\\
&\le\frac{8J\rme^{T^2}}{T^5}.
\label{eq:improved_integral_third}
\end{align}
Combining~\eqref{eq:improved_fejer_step}--
\eqref{eq:improved_integral_third} yields
\begin{equation}
\eta T\le \pi\sqrt J
+2J\left(\rme-1+\rme^{T^2/4}\right)
+\frac{8J\rme^{T^2}}{T^5}.
\label{eq:improved_integral_master}
\end{equation}

For $H\ge6$ choose
\begin{equation}
T^2=H+\frac12\log H.
\label{eq:improved_T_choice}
\end{equation}
In particular, $T\ge2$ and
\begin{equation}
H\le T^2\le \frac{23}{20}H.
\label{eq:improved_T_bounds}
\end{equation}
Suppose, for a contradiction, that
$J\le H\rme^{-2H}/\pi^2$. Substituting this bound into
\eqref{eq:improved_integral_master}, using
$\rme^{T^2}=\rme^H\sqrt H$, and dividing by $\eta T$ gives
\begin{align}
1\le{}&\sqrt{\frac{H}{T^2}}
+\frac{2H\rme^{-H}}{\pi^2T}
\left(\rme-1+\rme^{T^2/4}\right)+\frac{8H^{3/2}}{\pi^2T^6}.
\label{eq:improved_contradiction_start}
\end{align}
To bound the first term, the inequality $\sqrt{1-x}\le1-x/2$
and~\eqref{eq:improved_T_bounds} give
\begin{equation}
\sqrt{\frac H{T^2}}
\le1-\frac{\log H}{4T^2}
\le1-\frac{5\log6}{23H}.
\end{equation}
For the last term,
\begin{equation}
\frac{8H^{3/2}}{\pi^2T^6}
\le\frac8{\pi^2H^{3/2}}
\le\frac8{\pi^2\sqrt6}\frac1H.
\end{equation}
Finally, the middle term in~\eqref{eq:improved_contradiction_start}
is bounded by
\begin{align}
&\frac{2H\rme^{-H}}{\pi^2T}
\left(\rme-1+\rme^{T^2/4}\right)\nonumber\\
&\quad\le\frac{2\sqrt H}{\pi^2}
\left[(\rme-1)\rme^{-H}
+H^{1/8}\rme^{-3H/4}\right]\nonumber\\
&\quad\le\frac{C_6}{H},
\end{align}
where
\begin{equation}
C_6:=\frac{2\cdot6^{3/2}}{\pi^2}
\left[(\rme-1)\rme^{-6}+6^{1/8}\rme^{-9/2}\right]
<0.055.
\end{equation}
Here the last step follows because $H^{3/2}\rme^{-H}$ and
$H^{13/8}\rme^{-3H/4}$ are decreasing for $H\ge6$.
Consequently, the right-hand side of
\eqref{eq:improved_contradiction_start} is at most
\begin{equation}
1-\frac1H\left(\frac{5\log6}{23}
-\frac8{\pi^2\sqrt6}-C_6\right)<1,
\end{equation}
since the expression in parentheses is positive. This is a
contradiction. Hence $J\ge H\rme^{-2H}/\pi^2$ for $H\ge6$.
Together with~\eqref{eq:improved_constant_branch} and
\eqref{eq:improved_entropy_smallH}, this proves the scalar bound.

The vector extension follows by the same projection argument as in
Appendix~\ref{app:vector_projection}. For general $\epsilon>0$, apply the result at unit
noise variance to $\sqrt\epsilon\bfX+\bfZ$ and use the scaling of
Fisher information and the invariance of discrete R\'enyi entropy
under injective scaling. This concludes the proof.
\end{IEEEproof}

Theorem~\ref{thm:universal_fisher_bound} is stated in terms of the R\'enyi entropy of order two. In some applications, however, it is more convenient to work with R\'enyi entropies of other orders. Accordingly, for \(p\in[0,2]\), define
\begin{equation}
\beta_p(\bfX) := \exp(-2H_p(\bfX))
\max\left\{\frac{2}{\pi},c_{\rm F}H_p(\bfX)\right\}.
\label{eq:def_beta_p_new}
\end{equation}
with the convention that $\beta_p(\bfX)=0$ whenever $H_p(\bfX)=\infty$. The monotonicity of R\'enyi entropy then yields the following extension of Theorem~\ref{thm:universal_fisher_bound}.

\begin{corollary}\label{cor:renyi_family}
Under the assumptions of Theorem~\ref{thm:universal_fisher_bound}, for every
$p\in[0,2]$ and $\epsilon>0$,
\begin{equation}
\bfJ\left(
\bfX+\frac{1}{\sqrt{\epsilon}}\bfZ
\right)
\succeq
\epsilon\beta_p(\bfX)\bfI_n.
\label{eq:renyi_family_bound}
\end{equation}
\end{corollary}

\begin{IEEEproof}
R\'enyi entropy is nonincreasing in its order, so
$H_p(\bfX)\ge H_2(\bfX)$ for $p\le2$. Define
\begin{equation}
b(h):=\rme^{-2h}\max\left\{\frac2\pi,\frac{h}{\pi^2}\right\},
\qquad h\ge0.
\end{equation}
For $0\le h\le2\pi$, the maximum equals $2/\pi$ and $b(h)$
is decreasing. For $h\ge2\pi$, the maximum equals $h/\pi^2$
and $b(h)$ is also decreasing, since $h\rme^{-2h}$ is decreasing
for $h\ge1/2$. Thus $b$ is nonincreasing on $[0,\infty)$.
Theorem~\ref{thm:universal_fisher_bound} now implies
\begin{equation}
\bfJ\left(\bfX+\frac1{\sqrt\epsilon}\bfZ\right)
\succeq \epsilon b(H_2(\bfX))\bfI_n
\succeq \epsilon b(H_p(\bfX))\bfI_n,
\end{equation}
which proves the result. The convention at infinite entropy follows
by continuity.
\end{IEEEproof}

The finite-support consequence exposes the role of the logarithm most clearly.
If
\begin{equation}
K=|\supp(P_{\bfX})|,
\end{equation}
then $H_0(\bfX)=\log K$, and therefore
\begin{equation}
\bfJ\left(
\bfX+\frac{1}{\sqrt{\epsilon}}\bfZ
\right)
\succeq
\frac{\epsilon}{K^2}
\max\left\{
\frac{2}{\pi},
\frac{\log K}{\pi^2}
\right\}\bfI_n.
\label{eq:finite_support_loewner}
\end{equation}
Thus the least Fisher-information direction cannot decay faster than
$(\log K)/K^2$ up to constants.  This directional interpretation is important:
the exponent $K^{-2}$ is not a defect of the theorem, but the price of asking
for a lower bound in \emph{every} direction. For $n=1$, the bound in~\eqref{eq:finite_support_loewner}
improves on Stam's $K^{-2}$ bound if and only if
\begin{equation}
\frac{\log K}{\pi^2}>1,
\end{equation}
i.e., for $K>\rme^{\pi^2}$.

\subsection{Dimension-Aware Trace Bounds}
\label{sec:trace_bound}

Recall that the trace operator is monotone in positive-semidefinite order, which by using Theorem~\ref{thm:universal_fisher_bound} implies that
\begin{equation}
J\left(\bfX+\frac{1}{\sqrt{\epsilon}}\bfZ\right)
\ge \frac{\epsilon n}{\pi^2}
H_2(\bfX)\rme^{-2H_2(\bfX)}.
\label{eq:our_inequlaity_after_trace}
\end{equation}
On the other hand, by applying Stam's inequality we get that 
\begin{equation}
J\left(\bfX+\frac{1}{\sqrt{\epsilon}}\bfZ\right)\ge \epsilon n \exp\left(-\frac{2}{n}H(\bfX)\right).  \label{eq:Stam_entropy_ineq}
\end{equation}
Comparing the two exponents, we see that the exponent in \eqref{eq:our_inequlaity_after_trace}, which is a consequence of Theorem~\ref{thm:universal_fisher_bound}, has a worse dimensionality dependence than Stam's inequality.  

We now present an inequality that retains the Rényi entropy in the exponent while also exhibiting the appropriate dependence on the ambient dimension. We first present a general inequality in terms of $\alpha$-mutual information and then specialize it to discrete case with R\'enyi entropy.

\begin{theorem}[Dimension-aware trace bound]
\label{thm:dimension_aware}
Let $n\ge1$, let $\bfX\in\bbR^n$, and let
$\bfZ\sim\cN(0,\bfI_n)$ be independent of $\bfX$. 

Then
\begin{equation}
J\left(
\bfX+\frac{1}{\sqrt{\epsilon}}\bfZ
\right)
\ge
\begin{cases}
\displaystyle
a_1\epsilon
\exp\left(
-2I_0(\bfX;\epsilon)
\right),
& n=1,
\\[2ex]
\displaystyle
a_n\epsilon
\exp\left(
-\frac{2}{n}
I_{\frac{1}{r_n}}\left(\bfX;r_n\epsilon\right)
\right),
& n\ge2,
\end{cases}
\label{eq:dimension_aware_main}
\end{equation}
where $r_n$ and $a_n$ are given in~\eqref{eq:rn_definition} and~\eqref{eq:fisher_renyi_entropy_power_constant}, respectively.
\end{theorem}

\begin{IEEEproof}
Set
\begin{equation}
\bfY=\bfX+\frac{1}{\sqrt{\epsilon}}\bfZ.
\end{equation}
Since $\bfY$ is a Gaussian convolution,
Proposition~\ref{prop:fisher_renyi_entropy_power} gives
\begin{equation}
J(\bfY)
\ge
\frac{a_n}{N_{r_n}(\bfY)}.
\label{eq:dimension_aware_from_entropy_power}
\end{equation}

Suppose first that $n\ge2$. By the Gaussian representation of Sibson
mutual information in Proposition~\ref{prop:sibson_properties},
\begin{equation}
I_{\frac{1}{r_n}}(\bfX;r_n\epsilon)
=
h_{r_n}(\bfY)
-
h_{r_n}\left(
\frac{1}{\sqrt{\epsilon}}\bfZ
\right).
\label{eq:sibson_entropy_difference}
\end{equation}
Moreover,
\begin{equation}
h_{r_n}\left(
\frac{1}{\sqrt{\epsilon}}\bfZ
\right)
=
\frac{n}{2}\log\frac{2\pi}{\epsilon}
+
\frac{n}{2}\frac{\log r_n}{r_n-1}.
\end{equation}
Consequently,
\begin{equation}
\frac{1}{N_{r_n}(\bfY)}
=
\epsilon
\exp\left(
-\frac{2}{n}
I_{\frac{1}{r_n}}(\bfX;r_n\epsilon)
\right).
\label{eq:entropy_power_sibson_identity}
\end{equation}

For $n=1$, Proposition~\ref{prop:sibson_properties} gives
\begin{equation}
h_\infty\left(
\frac{1}{\sqrt{\epsilon}}Z
\right)
=
\frac{1}{2}\log\frac{2\pi}{\epsilon},
\end{equation}
the definition of $N_\infty$ yields
\begin{equation}
N_\infty(\bfY)
=
\frac{1}{2\pi}
\exp\left(2h_\infty(\bfY)\right),
\end{equation}
thus we obtain
\begin{equation}
\frac{1}{N_\infty(\bfY)}
=
\epsilon
\exp\left(
-2I_0(\bfX;\epsilon)
\right). \label{eq:entropy_power_scaled_order_zero_identity}
\end{equation}
Substituting \eqref{eq:entropy_power_sibson_identity} and
\eqref{eq:entropy_power_scaled_order_zero_identity} into
\eqref{eq:dimension_aware_from_entropy_power} proves the result.
\end{IEEEproof}

Combining Theorem~\ref{thm:dimension_aware} with the discrete-input bound
on Sibson mutual information in
Proposition~\ref{prop:sibson_properties} gives the following
R\'enyi-entropy consequence.

\begin{corollary}
\label{cor:dimension_aware_discrete}
Let $\bfX\in\bbR^n$ be a discrete random vector, and let
$\bfZ\sim\cN(0,\bfI_n)$ be independent of $\bfX$. Then, for every
$\epsilon>0$,
\begin{equation}
J\left(
\bfX+\frac{1}{\sqrt{\epsilon}}\bfZ
\right)
\ge
a_n\epsilon
\exp\left(
-\frac{2}{n}H_{r_n}(\bfX)
\right).
\label{eq:dimension_aware_discrete}
\end{equation}
\end{corollary}

We make the following remarks about Theorem~\ref{thm:dimension_aware} and Corollary~\ref{cor:dimension_aware_discrete}:
\begin{itemize}
    \item  The coefficient $a_n$  has the following properties 
(see Appendix~\ref{app:properties_an}):
\begin{subequations}
\begin{align}
0<a_n&<n,
\label{eq:an_less_n}
\\
\lim_{n\to\infty}\frac{a_n}{n}&=1,
\label{eq:an_ratio_limit}
\\
\lim_{n\to\infty}(n-a_n)&=1-\log 2.
\label{eq:an_difference_limit}
\end{align}
\label{eq:assumptotics_of_an}
\end{subequations}
\item One way to see that the bound in \eqref{eq:dimension_aware_main} can be strictly stronger than Stam's inequality is to compare their finiteness conditions. It was shown in \cite{milanian2026alpha} that
\begin{equation}
I_\alpha(\bfX;\epsilon)<\infty,
\qquad \epsilon>0,
\end{equation}
for every \(\alpha\in(0,1)\). On the other hand,
\begin{equation}
I(\bfX;\epsilon)<\infty
\quad\Longleftrightarrow\quad
H(\lfloor\bfX\rfloor)<\infty;
\end{equation}
see \cite{wu2011functional}, where the floor function is applied componentwise. Consequently, when
\(H(\lfloor\bfX\rfloor)=\infty\), Stam's inequality reduces to the trivial lower bound $J\ge0$, whereas the bound in \eqref{eq:dimension_aware_main} remains strictly positive.
\item The dimension-aware bound in \eqref{eq:dimension_aware_discrete} is
stronger than the entropy-based consequence of Stam's inequality in
\eqref{eq:Stam_entropy_ineq} whenever
\begin{equation}
H(\bfX)-H_{r_n}(\bfX)
>
\frac{n}{2}\log\frac{n}{a_n}.
\label{eq:dimension_aware_vs_stam_condition}
\end{equation}

\item In higher dimensions, the analogue of \eqref{eq:intro_K_stam} obtained by using
\begin{equation}
H_p(\bfX)\le H_0(\bfX)=\log K
\end{equation}
follows from both Stam's inequality in \eqref{eq:Stam_entropy_ineq} and the
dimension-aware bound in \eqref{eq:dimension_aware_discrete}.  However,
since $a_n<n$, Stam's inequality yields a slightly better
dimension-dependent constant when only the support cardinality $K$ is
used.  Indeed, the two bounds reduce respectively to
\begin{equation}
J\left(
\bfX+\frac{1}{\sqrt{\epsilon}}\bfZ
\right)
\ge
 \epsilon n K^{-2/n}
\end{equation}
and
\begin{equation}
J\left(
\bfX+\frac{1}{\sqrt{\epsilon}}\bfZ
\right)
\ge
\epsilon a_n K^{-2/n}.
\end{equation}
\item Finally, we remark on the scalar case of \eqref{eq:dimension_aware_discrete}. For $n=1$, the bound reduces to
\begin{equation}
J\left(
X+\frac{1}{\sqrt{\epsilon}}Z
\right)
\ge
\frac{2\epsilon}{\pi}
\rme^{-2H_\infty(X)}
=\frac{2\epsilon}{\pi}
\max_j p_j^2.
\label{eq:max_prob_fisher}
\end{equation}
Since $p\mapsto H_p(X)$ is nonincreasing, $H_\infty(X)$ yields the largest exponential factor among all R\'enyi entropy orders. 
Thus, a small normalized Fisher information forces the probability mass of $X$ to be dispersed among its support points. In particular, when the support size is fixed, this bound provides a quantitative indication that the distribution of $X$ is close to uniform.
 \item The results in this work can also be viewed through the lens of best approximation theory in terms of Fisher divergence. In particular, recall that the Fisher divergence is defined as 
 \begin{equation}
     \cI(P\|Q) = \int | \rho_P(x) -\rho_Q(x) |^2 f_X(x) \rmd x,
 \end{equation}
 where $f_X$ is the density of distribution $P$ and where $\rho_P$ and $\rho_Q$ are score functions of $P$ and $Q$, respectively. Then, our results can be restated into the following form: for $v>0$
 \begin{align}
    & \min_{P_X: \: X  \text{ discrete},\, \bbE[X^2]=v,\,\bbE[X]=0} \cI(P_{X+Z}\| P_{X_G+Z})  \ge \left( \beta_p(X) - (v+1) \right )_{+},  \label{eq:lower_bound_fisher_div}
 \end{align}
 where $X_G \sim \mathcal{N}(0,v)$. The lower bound in \eqref{eq:lower_bound_fisher_div} follows by using that
 \begin{equation}
     \cI(P_{X+Z}\| P_{X_G+Z}) = J(X+Z) - J(X_G+Z).
 \end{equation}
 \end{itemize}

The dimension-aware bound in Theorem~\ref{thm:dimension_aware} depends on
the mixing distribution through Sibson mutual information of order
$1/r_n$. The following complementary bound uses the order $1/2$ Sibson
mutual information. 

\begin{theorem}
\label{thm:gaussian_kernel_fisher}
Let $n\ge4$, let $\bfX$ be an arbitrary random vector in $\bbR^n$, and let
$\bfZ\sim\cN(0,\bfI_n)$ be independent of $\bfX$. Then, for every
$\epsilon>0$,
\begin{equation}
J\left(\bfX+\frac{1}{\sqrt{\epsilon}}\bfZ\right)
\ge\epsilon a_n\exp\left(-\frac{1}{2}I_{\frac12}\left(\bfX;\frac{8\epsilon}{n}\right)\right).
\label{eq:Gaussian_kernel_fisher_sibson}
\end{equation}
\end{theorem}

\begin{IEEEproof}
Let $f_u$ denote the density of $\bfX+u\bfZ$. We first note the identity
\begin{equation}
\|f_u\|_2^2
=
(4\pi u^2)^{-n/2}
\bbE\left[
\exp\left(
-\frac{\|\bfX-\bfX'\|^2}{4u^2}
\right)
\right],
\label{eq:L2_gaussian_collision}
\end{equation}
where $\bfX'$ is an independent copy of $\bfX$. Indeed, this follows by
conditioning on $(\bfX,\bfX')$ and using
\begin{align}
&\int_{\bbR^n}
\phi_u(\bfy-\bfX)\phi_u(\bfy-\bfX')\,\rmd\bfy
\nonumber\\
&\quad=\frac{1}{
(4\pi u^2)^{\frac{n}{2}}}
\exp\left(
-\frac{\|\bfX-\bfX'\|^2}{4u^2}
\right).
\end{align}

Assume first that $n>4$ and set
\begin{equation}
p=\frac{n}{n-2},
\quad
q=\frac{2n}{n+4},
\quad
s=\sqrt{\frac{n}{4\epsilon}},
\quad
\sigma=\sqrt{\frac{n-4}{4\epsilon}}.
\label{eq:young_parameter_choice}
\end{equation}
Then
\begin{equation}
s^2=\frac{1}{\epsilon}+\sigma^2,
\qquad
f_s=f_{1/\sqrt{\epsilon}}*\phi_\sigma,
\end{equation}
and \(1/p+1/q=1+1/2\). The sharp Young inequality
\eqref{eq:sharp_young_prelim} therefore gives
\begin{equation}
\|f_s\|_2
\le
(C_pC_q)^n
\|f_{1/\sqrt{\epsilon}}\|_p
\|\phi_\sigma\|_q.
\label{eq:young_kernel_step}
\end{equation}
Using
\begin{equation}
\|\phi_\sigma\|_q
=
(2\pi\sigma^2)^{-\frac{n}{2}(1-1/q)}
q^{-n/(2q)}
\end{equation}
together with \eqref{eq:L2_gaussian_collision}, and simplifying the
sharp Young constants, yields
\begin{align}
\|f_{1/\sqrt{\epsilon}}\|_p
\ge{}&
\frac{\epsilon}{2\pi}
\left(\frac{n-2}{n}\right)^{\frac{n-2}{2}}
\left(
\bbE \left[\exp\left(
-\frac{\epsilon}{n}\|\bfX-\bfX'\|^2
\right) \right]
\right)^{\frac{1}{2}}.
\label{eq:young_simplified}
\end{align}
By Proposition~\ref{prop:sibson_properties},
\begin{equation}
\bbE \left[\exp\left(
-\frac{\epsilon}{n}\|\bfX-\bfX'\|^2
\right) \right]
=
\exp\left(
-I_{\frac12}\left(\bfX;\frac{8\epsilon}{n}\right)
\right).
\label{eq:kernel_sibson_conversion}
\end{equation}
Finally, the sharp Sobolev inequality gives
\begin{equation}
J\left(
\bfX+\frac{1}{\sqrt{\epsilon}}\bfZ
\right)
\ge
n(n-2)|\mathbb S^n|^{2/n}
\|f_{1/\sqrt{\epsilon}}\|_{n/(n-2)}.
\end{equation}
Combining this inequality with
\eqref{eq:young_simplified}--\eqref{eq:kernel_sibson_conversion} and the
definition of $a_n$ proves \eqref{eq:Gaussian_kernel_fisher_sibson}.

When $n=4$, the additional smoothing disappears, since $s=1/\sqrt{\epsilon}$
and $p=2$. Thus, \eqref{eq:L2_gaussian_collision} followed directly by the
sharp Sobolev inequality gives the same bound.
\end{IEEEproof}

Theorem~\ref{thm:gaussian_kernel_fisher} matches Theorem~\ref{thm:dimension_aware} for $n=4$ and is looser for $n>4$.  We keep this bound for the only reason that mutual information of order $\alpha=1/2$ can be easier to compute because 
\begin{equation}
    I_{\frac{1}{2}}\left(\bfX; \frac{8\epsilon}{n} \right) = -\log \bbE[\rme^{-\epsilon \|\bfX-\bfX' \|^2/n}].
\end{equation}

\section{Applications}\label{sec:applications}
We now illustrate how new Fisher information inequalities can be used.

\subsection{Rotationally Invariant Discrete-Shell Mixtures}
\label{subsec:bounds_for_shells}
We now specialize our bounds to rotationally invariant distributions supported
on countably many concentric spheres. Such distributions appear as extremizers of mutual information under ball constraints on the support~\cite{vector_Gausss}, rate-distortion problems~\cite{dytso2024uniform}, and estimation problems~\cite{berry1990minimax,marchand2002minimax}.  The goal is to find an upper bound on the $\alpha$-mutual information and then combine it with Theorem~\ref{thm:dimension_aware}.  To present a bound on the $\alpha$-mutual information we need the following covering bound.

\begin{lemma}[Partition of the sphere]
\label{lem:sphere_partition}
Let $n\ge2$ and $\delta\in(0,1]$. Then there exists a measurable
partition
\begin{equation}
\mathbb S^{n-1}
=
\mathcal C_1\cup\cdots\cup\mathcal C_N
\end{equation}
such that
\begin{equation}
N\le M_n\delta^{-(n-1)}
\label{eq:sphere_partition_number}
\end{equation}
and
\begin{equation}
\sup_{\bfu,\bfv\in\mathcal C_k}
\|\bfu-\bfv\|
\le\delta,
\qquad k=1,\ldots,N,
\label{eq:sphere_partition_diameter}
\end{equation}
where
\begin{equation}
M_n
:=
4^{n-1}(n-1)
\frac{|\mathbb S^{n-1}|}{|\mathbb S^{n-2}|}
\left(\frac{\pi}{2}\right)^{n-2}.
\label{eq:def_sphere_covering_constant}
\end{equation}
\end{lemma}

\begin{IEEEproof}
The existence of diameter-bounded partitions of the sphere follows from
the general partition results in \cite{gigante2017diameter}. For
completeness, and to obtain an explicit constant needed below, we provide
a direct packing argument in Appendix~\ref{app:sphere_partition}.
\end{IEEEproof}

We now use Lemma~\ref{lem:sphere_partition}  to provide the following bound on the $\alpha$-mutual information. 

\begin{lemma}
\label{lem:sibson_shell_intrinsic}
Let $n\ge2$, let $S$ be a discrete random variable supported on $[0,A]$,
and let $\bfU$ be uniformly distributed on $\mathbb S^{n-1}$
independently of $S$.  Then, for every $\epsilon>0$
and $\alpha\in(0,1]$,
\begin{equation}
I_\alpha(S\bfU;\epsilon) \le H_{1/\alpha}(S)+(n-1)\log\left(1+A\sqrt{\epsilon}\right)+ \xi_n,
\label{eq:sibson_shell_intrinsic}
\end{equation}
where 
\begin{equation}
\xi_n:=\log M_n+\frac12,
\label{eq:def_xi_n}
\end{equation}
and where $M_n$ is given in \eqref{eq:def_sphere_covering_constant}.
\end{lemma}

\begin{IEEEproof}
We first consider a single shell of radius $s$. If
$s\sqrt{\epsilon}\le1$, monotonicity in the order and the Gaussian
maximum-entropy bound give
\begin{align}
I_\alpha(s\bfU;\epsilon)
&\le I(s\bfU;s\bfU+\epsilon^{-1/2}\bfZ)\\
&\le \frac{n}{2}
\log\left(1+\frac{\epsilon s^2}{n}\right)
\le \frac12.
\label{eq:bound_for_small_eps}
\end{align}

Suppose next that $s\sqrt{\epsilon}>1$ and set
\begin{equation}
\delta=\frac{1}{s\sqrt{\epsilon}}.
\end{equation}
By Lemma~\ref{lem:sphere_partition}, there exists a measurable
partition
\begin{equation}
\mathbb S^{n-1}
=
\mathcal C_1\cup\cdots\cup\mathcal C_N
\end{equation}
such that
\begin{equation}
N\le M_n\delta^{-(n-1)}
\end{equation}
and
\begin{equation}
\sup_{\bfu,\bfv\in\mathcal C_k}
\|\bfu-\bfv\|
\le\delta,
\qquad k=1,\ldots,N.
\end{equation}
Define the discrete random variable $V$ by
\begin{equation}
V=k
\quad\Longleftrightarrow\quad
\bfU\in\mathcal C_k.
\end{equation}
Thus, $V$ specifies the cell containing $\bfU$. In particular,
conditional on $V=k$, any two possible values of $s\bfU$ are separated
by at most
\begin{equation}
s\delta=\frac{1}{\sqrt{\epsilon}}.
\end{equation}
Moreover, since $V$ takes at most $N$ values,
\begin{equation}
H_{1/\alpha}(V)
\le\log N
\le
\log M_n+(n-1)\log\frac{1}{\delta}.
\label{eq:entropy_partition_index}
\end{equation}

Since $V$ is a deterministic function of $s\bfU$, we have
\begin{equation}
I_\alpha(s\bfU;\epsilon)
=
I_\alpha(V,s\bfU;\epsilon).
\end{equation}
Applying Lemma~\ref{lem:sibson_mixture} with the cell index $V$ as the
mixing variable gives
\begin{equation}
I_\alpha(s\bfU;\epsilon)
\le
H_{1/\alpha}(V)
+
\sup_{v\in\supp(P_V)}
I_\alpha(s\bfU;\epsilon\mid V=v).
\label{eq:sibson_sphere_partition}
\end{equation}
Conditioned on $V=v$, the support of $s\bfU$ has diameter at most
$\epsilon^{-1/2}$. Consequently,
\begin{equation}
\Tr\left(\Cov(s\bfU\mid V=v)\right)
\le\frac{1}{\epsilon}.
\end{equation}
Hence, using a bound similar to \eqref{eq:bound_for_small_eps},
\begin{align}
I_\alpha(s\bfU;\epsilon\mid V=v)
&\le
I(s\bfU;s\bfU+\epsilon^{-1/2}\bfZ\mid V=v)
\\
&\le
\frac{1}{2}
\log\det\left(
\bfI_n+\epsilon\Cov(s\bfU\mid V=v)
\right)
\\
&\le
\frac{\epsilon}{2}
\Tr \left(\Cov(s\bfU\mid V=v) \right)
\le\frac12.
\end{align}
Moreover, since $\delta=(s\sqrt{\epsilon})^{-1}$,
\eqref{eq:entropy_partition_index} becomes
\begin{equation}
H_{1/\alpha}(V)
\le
\log M_n+(n-1)\log(s\sqrt{\epsilon}).
\end{equation}
It follows that, when $s\sqrt{\epsilon}>1$,
\begin{align}
I_\alpha(s\bfU;\epsilon)
&\le
(n-1)\log(s\sqrt{\epsilon})
+\log M_n+\frac12\\
&\le
(n-1)\log\left(1+s\sqrt{\epsilon}\right)+\xi_n.
\end{align}
When $s\sqrt{\epsilon}\le1$, the same bound follows from
\eqref{eq:bound_for_small_eps} and $\xi_n\ge1/2$. Therefore, for every
$s\ge0$,
\begin{equation}
I_\alpha(s\bfU;\epsilon)
\le
(n-1)\log\left(1+s\sqrt{\epsilon}\right)+\xi_n.
\label{eq:single_shell_intrinsic_bound}
\end{equation}

Finally, applying Lemma~\ref{lem:sibson_mixture} with the shell radius
$S$ as the mixing variable gives
\begin{align}
I_\alpha(S\bfU;\epsilon)
&\le
H_{1/\alpha}(S)
+
\sup_{s\in\supp(P_S)}
I_\alpha(s\bfU;\epsilon)\\
&\le
H_{1/\alpha}(S)
+
(n-1)\log\left(1+A\sqrt{\epsilon}\right)
+
\xi_n,
\end{align}
where the last inequality follows from $S\le A$ and
\eqref{eq:single_shell_intrinsic_bound}. This proves
\eqref{eq:sibson_shell_intrinsic}.
\end{IEEEproof}

\begin{rem}
Using the standard gamma-ratio asymptotics, one can show that 
\begin{equation}
\frac{\xi_n}{n}
=
\log(2\pi)
+
O\left(\frac{\log n}{n}\right),
\end{equation}
and hence
\begin{equation}
\lim_{n\to\infty}\frac{\xi_n}{n}
=
\log(2\pi).
\label{eq:xi_n_asymptotic}
\end{equation}

\end{rem}

Combining Lemma~\ref{lem:sibson_shell_intrinsic} with
Theorem~\ref{thm:dimension_aware} gives the following Fisher-information
bound for mixtures of spherical shells.

\begin{theorem}[Dimension-aware shell bound]
\label{thm:shell_sobolev}
Let $n\ge2$, let $R$ be a discrete nonnegative random variable supported
on $[0,A]$, and let $\bfU$ be uniformly distributed on
$\mathbb S^{n-1}$ independently of $R$. Then, for every $\epsilon>0$,
\begin{align}
J\left(
R\bfU+\frac{1}{\sqrt{\epsilon}}\bfZ
\right)
\ge{}&
a_n\epsilon
\exp\left(-\frac{2\xi_n}{n}\right)
\left(1+A\sqrt{r_n\epsilon}\right)^{-\frac{2(n-1)}{n}}
\exp\left(
-\frac{2}{n}H_{r_n}(R)
\right).
\label{eq:fisher_shell_entropy_bound}
\end{align}
\end{theorem}

\begin{IEEEproof}
Applying Lemma~\ref{lem:sibson_shell_intrinsic} with $\alpha=\frac{1}{r_n}$ and $\epsilon'=r_n\epsilon$ gives
\begin{align}
I_{\frac{1}{r_n}}(R\bfU;r_n\epsilon)
\le{}&
H_{r_n}(R)
+
(n-1)\log\left(1+A\sqrt{r_n\epsilon}\right)
+
\xi_n.
\end{align}
Substituting this inequality into
Theorem~\ref{thm:dimension_aware} yields
 \eqref{eq:fisher_shell_entropy_bound}.
\end{IEEEproof}

Finally, we specialize the above bound to the number of spherical shells. 
\begin{corollary}[Finite-shell bound]
\label{cor:fisher_finite_shells}
Under the assumptions of Theorem~\ref{thm:shell_sobolev}, suppose that $K=\left|\supp(P_R)\right|<\infty.$ Then, for every $\epsilon>0$,
\begin{equation}
J\left(
R\bfU+\frac{1}{\sqrt{\epsilon}}\bfZ
\right)
\ge
a_n\epsilon
 \rme^{-\frac{2\xi_n}{n}}
\left(1+A\sqrt{r_n\epsilon}\right)^{-\frac{2(n-1)}{n}}
K^{-  \frac{2}{n}  }.
\label{eq:fisher_finite_shell_count}
\end{equation}
\end{corollary}

\begin{IEEEproof}
Since $R$ is supported on $K$ values,
\begin{equation}
H_{r_n}(R)\le\log K.
\end{equation}
The result now follows directly from
\eqref{eq:fisher_shell_entropy_bound}.
\end{IEEEproof}

We will use \eqref{eq:fisher_finite_shell_count} in
Section~\ref{sec:shell_lower_bound} to obtain a lower bound on the number
of shells in the least favorable distribution for the bounded normal
mean problem.

\subsection{Bounded Normal Mean: Scalar Case}
\label{sec:bounded_mean_problem}
We now consider the following estimation problem.  Let
\begin{equation}
Y=\theta+Z,
\qquad Z\sim\cN(0,1),
\qquad |\theta|\le R,
\end{equation}
with squared-error loss. This problem is known as estimation of a bounded Gaussian and the interested reader is referred to \cite{casella1981estimating,berry1990minimax,marchand2004estimation} and the references therein for an introduction to this problem. The corresponding minimax risk is
\begin{equation}
\inf_{\widehat\theta}
\sup_{|\theta|\le R}
\bbE_\theta\left[ \left(\widehat\theta(Y)-\theta \right)^2\right].
\label{eq:scalar_minimax_risk}
\end{equation}
By the standard minimax-Bayes duality for this problem
(see, e.g., \cite{marchand2004review}), \eqref{eq:scalar_minimax_risk} equals
\begin{equation}
\sup_{P_X:\,|X|\le R\ \text{almost surely}}\mmse(X|X+Z).
\label{eq:least_favorable_scalar_problem}
\end{equation}
A maximizing distribution in \eqref{eq:least_favorable_scalar_problem} denoted as $P_{X^\star}$ is called least favorable prior.  The distribution $P_{X^\star}$ is known to be 
unique, symmetric, and finitely discrete \cite{yagli2019estimation}. 

Our goal is to quantify properties of $P_{X^\star}$. In particular, next, by using the newly derived inequalities we establish lower bounds on the support size and upper bounds on the probability values. 

We begin by considering the following distribution with density given by
\begin{equation}
f_{\widetilde X}(x)
=
\frac1R\cos^2\left(\frac{\pi x}{2R}\right)
\mathbf 1_{\{|x|\le R\}}.
\label{eq:cosine_trial_prior}
\end{equation}
The distribution in \eqref{eq:cosine_trial_prior} is an analog of the uniform distribution for the continuous entropy in a sense that it minimizes Fisher information under a support constraint (i.e.,   $X \in [-R,R]$) \cite{bickel1981minimax,bercher2009minimum} and results in Fisher information equal to
\begin{equation}
J(\widetilde X)=\frac{\pi^2}{R^2}.
\label{eq:trial_prior_fisher}
\end{equation}
Additionally, in \cite{bickel1981minimax}, it was shown that the least favorable distribution $P_{X^\star}$ converges to $P_{\widetilde X}$ in the sense that $ \frac{1}{R}X^\star  \to \frac{1}{R} \widetilde{X}$ in distribution.

\begin{theorem}
\label{thm:least_favorable_scalar}
Let $P_{X^\star}$ maximize \eqref{eq:least_favorable_scalar_problem} and let $K:=|\supp(P_{X^\star})|$.
Then, the following statements hold: 
\begin{enumerate}
    \item  for every $R>0$,
    \begin{equation}
\hspace{-0.3cm}K\ge\max\Bigg\{
\sqrt{1+\frac{R^2}{\pi^2}},
\sqrt{
\frac{c_{\rm F}}{2}
\left(1+\frac{R^2}{\pi^2}\right)
\log\left(1+\frac{R^2}{\pi^2}\right)
}
\Bigg\}.
\label{eq:least_favorable_scalar_bound}
\end{equation}
In particular,
\begin{equation}
K
\ge
\frac{1}{\pi^2}
R\sqrt{\log R}\,(1+o(1)),
\qquad R\to\infty.
\label{eq:least_favorable_scalar_asymptotic}
\end{equation}
\item For every $R>0$ and $x \in \supp(P_{X^\star})$
\begin{equation}
    P_{X^\star}(x)  \le \frac{\pi^{3/2}}{\sqrt{2(R^2+\pi^2)}}.
\end{equation}
\end{enumerate}

\end{theorem}

\begin{IEEEproof}
Applying the  Bayesian Cram\'er-Rao inequality to the MMSE with the input distribution in \eqref{eq:cosine_trial_prior}, we get 
\begin{equation}
\mmse(\widetilde X|\widetilde X+Z)
\ge
\frac{1}{1+J(\widetilde X)}
=
\frac{R^2}{R^2+\pi^2};
\label{eq:trial_prior_bcrb}
\end{equation}
see, e.g., \cite{vantrees2001detection}.  Since $P_{X^\star}$ is least
favorable,
\begin{equation}
\mmse(X^\star|X^\star+Z)
\ge
\frac{R^2}{R^2+\pi^2}.
\end{equation}
Brown's identity in \eqref{eq:brown_identity}, then gives
\begin{equation}
J(X^\star+Z)
=1-\mmse(X^\star|X^\star+Z)
\le
\frac{\pi^2}{R^2+\pi^2}.
\label{eq:least_fav_output_fisher_upper}
\end{equation}

The entropy consequence of Stam's inequality in
\eqref{eq:intro_stam_entropy} gives
\begin{equation}
J(X^\star+Z)
\ge \rme^{-2H(X^\star)}\ge K^{-2}.
\end{equation}
Combining with \eqref{eq:least_fav_output_fisher_upper},
\begin{equation}
K\ge\sqrt{1+\frac{R^2}{\pi^2}}.
\label{eq:least_fav_linear_prebound}
\end{equation}
On the other hand, Corollary~\ref{cor:renyi_family} with $p=0$ gives
\begin{equation}
J(X^\star+Z)
\ge
c_{\rm F}\frac{\log K}{K^2}.
\end{equation}
Together with \eqref{eq:least_fav_output_fisher_upper},
\begin{equation}
\frac{\pi^2}{R^2+\pi^2}
\ge
c_{\rm F}\frac{\log K}{K^2}.
\label{eq:least_fav_intermediate}
\end{equation}
Finally, \eqref{eq:least_fav_linear_prebound} implies
\begin{equation}
\log K
\ge
\frac12\log\left(1+\frac{R^2}{\pi^2}\right).
\end{equation}
Substituting in \eqref{eq:least_fav_intermediate} proves
\eqref{eq:least_favorable_scalar_bound}; the asymptotic statement follows
immediately.

To prove the second statement, observe that \eqref{eq:max_prob_fisher} implies
\begin{equation}
\max_x P_{X^\star}(x)
\le
\sqrt{\frac{\pi}{2}J(X^\star+Z)}.
\end{equation}
Combining this inequality with \eqref{eq:least_fav_output_fisher_upper} gives
\begin{equation}
\max_x P_{X^\star}(x)
\le
\sqrt{
\frac{\pi^3}{2(R^2+\pi^2)}
}
=
\frac{\pi^{3/2}}{\sqrt{2(R^2+\pi^2)}},
\end{equation}
which proves the second statement.
This concludes the proof.
\end{IEEEproof}

The $R\sqrt{\log R}$ branch of
\eqref{eq:least_favorable_scalar_bound} exceeds its linear branch
exactly when
\begin{equation}
R>\pi\sqrt{\exp(2\pi^2)-1}
\approx6.07\cdot10^4.
\end{equation}

The proof also explains why the logarithmic improvement in the scalar Fisher
bound matters.  The classical $K^{-2}$ estimate and the minimax-risk upper
bound both yield only lower bounds of linear order in $R$.  The extra $\log K$
forces the support to grow strictly faster than linearly. 

\begin{rem}
The $R\sqrt{\log R}$ support-size lower bound in Theorem~\ref{thm:least_favorable_scalar} parallels a recent result in \cite{Wang_AMplitude}, where the same lower bound was obtained for the support size of an input distribution maximizing mutual information under a peak-amplitude constraint. 
 \end{rem}

\subsection{Bounded Normal Mean: Vector Case}
\label{sec:shell_lower_bound}

We now consider the multivariate counterpart of the bounded normal mean
problem. Let the unknown mean belong to the Euclidean ball
\begin{equation}
B_R:=\{\bfx\in\bbR^n:\|\bfx\|\le R\},
\end{equation}
and let the loss be squared Euclidean error. Equivalently, a least
favorable prior maximizes the Bayes risk
\begin{equation}
\sup_{P_{\bfX}:\,\bfX\in B_R\ \mathrm{a.s.}}
\mmse(\bfX\mid\bfX+\bfZ),
\qquad
\bfZ\sim\cN(0,\bfI_n).
\label{eq:least_favorable_vector_problem}
\end{equation}
Let $P_{\bfX_R^\star}$ be a maximizing prior. It is known to be unique,
spherically symmetric, and supported on finitely many concentric spheres
\cite{berry1990minimax},\cite[Prop.~4]{dytso2018structure}. We may therefore write
\begin{equation}
\bfX_R^\star=S_R^\star\bfU,
\end{equation}
where $\bfU$ is uniformly distributed on $\mathbb S^{n-1}$ independently
of the discrete radial component
\begin{equation}
S_R^\star:=\|\bfX_R^\star\|.
\end{equation}
Define
\begin{equation}
K_n(R):=
\left|\supp(P_{S_R^\star})\right|
\label{eq:def_shell_count}
\end{equation}
to be the number of shells.

As in the scalar case, the proof proceeds by comparing an upper bound on the Fisher information of the least favorable output with a lower bound that depends on the number of occupied shells. We first introduce a continuous  prior on $B_R$, which is the multivariate analogue of the cosine-squared prior in \eqref{eq:cosine_trial_prior}.

Let $J_\nu$ denote the Bessel function of the first kind and let $j_{\nu,1}$ be its first positive zero. Bounds on the value of $j_{\nu,1}$ can be found in \cite{ismail1995bounds,qu1999best} where it was shown that
\begin{equation}
    j_{\nu,1} = \nu +O \left(\nu^{1/3} \right).
\end{equation}

Define
\begin{equation}
\nu=\frac{n}{2}-1
\end{equation}
and let $\widetilde{\bfX}$ have density
\begin{equation}
p_R(\bfx)
=
\frac{2}
{|\mathbb S^{n-1}|R^2J_{\nu+1}^2(j_{\nu,1})}
\frac{
J_{\nu}^2\left(
j_{\nu,1}\frac{\|\bfx\|}{R}
\right)
}
{\|\bfx\|^{n-2}}
\mathbf 1_{\{\|\bfx\|\le R\}},
\label{eq:vector_trial_prior}
\end{equation}
where the value at $\bfx=\mathbf 0$ is defined by continuity. 
In dimensions $n=1,2,3$, the density in
\eqref{eq:vector_trial_prior} takes the explicit forms
\begin{equation}
p_R(\bfx)
=
\begin{cases}
\displaystyle
\frac{1}{R}
\cos^2\left(\frac{\pi x}{2R}\right)
\mathbf 1_{\{|x|\le R\}},
& n=1,
\\[2ex]
\displaystyle
\frac{
J_0^2\left(j_{0,1}\frac{\|\bfx\|}{R}\right)
}{
\pi R^2J_1^2(j_{0,1})
}
\mathbf 1_{\{\|\bfx\|\le R\}},
& n=2,
\\[2ex]
\displaystyle
\frac{1}{2\pi R\|\bfx\|^2}
\sin^2\left(\frac{\pi\|\bfx\|}{R}\right)
\mathbf 1_{\{\|\bfx\|\le R\}},
& n=3,
\end{cases}
\label{eq:vector_trial_prior_low_dimensions}
\end{equation}

The distribution in  \eqref{eq:vector_trial_prior} minimizes 
Fisher information subject to the ball constraint \cite[Thm.~4]{dehesa2012generalized} and the resulting Fisher information is given by 
\begin{equation}
J(\widetilde{\bfX})
=
\frac{4j_{\nu,1}^2}{R^2}.
\label{eq:minimum_fisher_ball}
\end{equation}
Additionally, we have that   $p_R$ is rotationally invariant and
\begin{equation}
\bfJ(\widetilde{\bfX})
=
\frac{4j_{\nu,1}^2}{nR^2}\bfI_n.
\label{eq:trial_prior_matrix_fisher}
\end{equation}

The distribution in \eqref{eq:vector_trial_prior} is plotted in Fig.~\ref{fig:least_fave_n_2} for $n=2$. 
\begin{figure}
    \centering
    \includegraphics[width=7cm]{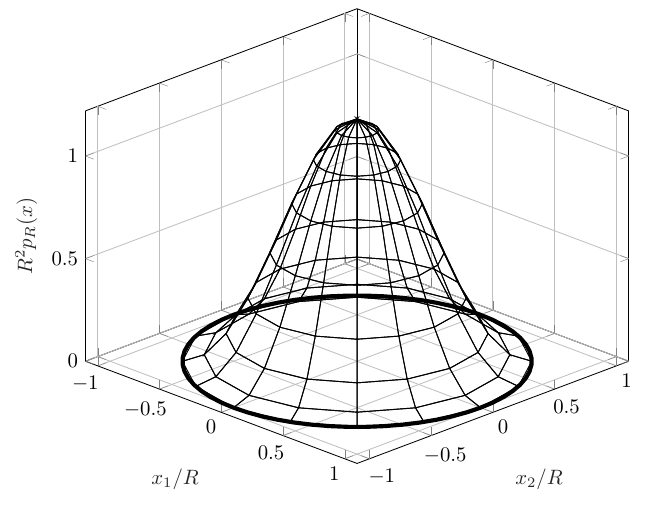}
    \caption{Density in \eqref{eq:vector_trial_prior} for $n=2$. The normalized coordinates $x_i/R$ and $R^2p_R(x)$ make the surface independent of $R$; the figure uses $R=1$.}
    \label{fig:least_fave_n_2}
\end{figure}

We now present the following lower bound on the number of shells. 

\begin{theorem}
\label{thm:bounded_shell_count}
For every \(n\ge2\) and \(R>0\), the number of occupied shells of the least
favorable prior satisfies
\begin{equation}
K_n(R)
\ge
\exp(-\xi_n)
\left(\frac{a_n}{U_n(R)}\right)^{n/2}
\left(1+\sqrt{r_n}R\right)^{-(n-1)}.
\label{eq:bounded_shell_explicit}
\end{equation}
where $a_n$ is defined in \eqref{eq:fisher_renyi_entropy_power_constant}, $r_n$ defined in \eqref{eq:rn_definition}, $\xi_n$ defined in \eqref{eq:def_xi_n}, respectively, and 
\begin{equation}
U_n(R)
:=
\frac{4j_{\nu,1}^2/R^2}
{1+4j_{\nu,1}^2/(nR^2)}.
\label{eq:def_Un_bounded}
\end{equation}
Consequently,  
\begin{equation}
K_n(R)
\ge
c_nR\left(1+o(1)\right),
\qquad R\to\infty,
\label{eq:bounded_shell_linear}
\end{equation}
where
\begin{equation}
c_n
:=
\frac{
a_n^{n/2}\exp(-\xi_n)
}{
\left(4j_{\nu,1}^2\right)^{n/2}
r_n^{(n-1)/2}
}.
\label{eq:def_cn_shell_growth}
\end{equation}
\end{theorem}

\begin{rem}
    The constant $c_n$ can be small and the bound on $K_n(R)$   can be inactive for large interval of $R$. However, to the best of our knowledge, this is the only available explicit bound in the literature, and we expect the constants to be refined in future works.  
\end{rem}

\begin{IEEEproof}
Consider the random vector \(\widetilde{\bfX}\) with density \(p_R\) in
\eqref{eq:vector_trial_prior}. By
\eqref{eq:trial_prior_matrix_fisher},
\begin{equation}
\bfJ(\widetilde{\bfX})
=
\frac{4j_{\nu,1}^2}{nR^2}\bfI_n.
\end{equation}
The matrix Bayesian Cram\'er-Rao inequality
\cite{vantrees2001detection} therefore gives
\begin{align}
{\bf MMSE}(\widetilde{\bfX}\mid\widetilde{\bfX}+\bfZ)
&\succeq
\left(\bfI_n+\bfJ(\widetilde{\bfX})\right)^{-1}\\
&=
\frac{1}
{1+4j_{\nu,1}^2/(nR^2)}
\bfI_n.
\label{eq:vector_trial_bcrb}
\end{align}
Since \(\bfX_R^\star\) is least favorable, its scalar MMSE is at least
that of \(\widetilde{\bfX}\). Hence,
\begin{equation}
\mmse(\bfX_R^\star\mid\bfX_R^\star+\bfZ)
\ge
\frac{n}
{1+4j_{\nu,1}^2/(nR^2)}.
\end{equation}
Brown's identity then yields
\begin{align}
J(\bfX_R^\star+\bfZ)
&=
n-\mmse(\bfX_R^\star\mid\bfX_R^\star+\bfZ)\\
&\le
n-
\frac{n}
{1+4j_{\nu,1}^2/(nR^2)}\\
&=
U_n(R).
\label{eq:bounded_vector_fisher_upper}
\end{align}

On the other hand, Corollary~\ref{cor:fisher_finite_shells}, applied
with \(\epsilon=1\), \(A=R\), and \(K=K_n(R)\), gives
\begin{equation}
J(\bfX_R^\star+\bfZ)
\ge
a_n
\rme^{-\frac{2\xi_n}{n}}
\left(1+\sqrt{r_n}R\right)^{-\frac{2(n-1)}{n}}
K_n(R)^{-2/n}.
\label{eq:bounded_vector_fisher_lower}
\end{equation}
Combining \eqref{eq:bounded_vector_fisher_upper} and
\eqref{eq:bounded_vector_fisher_lower}, and solving for $K_n(R)$, proves \eqref{eq:bounded_shell_explicit}.

Finally,
\begin{equation}
U_n(R)=4j_{\nu,1}^2R^{-2}\left(1+o(1)\right)
\end{equation}
and
\begin{equation}
\left(1+\sqrt{r_n}R\right)^{n-1}
=r_n^{(n-1)/2}R^{n-1}\left(1+o(1)\right).
\end{equation}
Substituting these asymptotic expressions into
\eqref{eq:bounded_shell_explicit} gives
\begin{equation}
K_n(R)\ge\frac{a_n^{n/2}\exp(-\xi_n)}{
\left(4j_{\nu,1}^2\right)^{n/2}r_n^{(n-1)/2}}
R\left(1+o(1)\right),
\end{equation}
which proves \eqref{eq:bounded_shell_linear} and
\eqref{eq:def_cn_shell_growth}.
\end{IEEEproof}

\subsection{Upper Bounds on the MMSE}
\label{sec:mmse_upper_bounds}

The usual linear minimum mean-square error (LMMSE) bound depends only on
covariance.  If
$\bfK_{\bfX}=\operatorname{Cov}(\bfX)$, then
\begin{equation}
{\bf MMSE}(\bfX|\bfY_\gamma)
\preceq
\bfK_{\bfX}(\bfI_n+\gamma\bfK_{\bfX})^{-1}.
\label{eq:lmmse_application}
\end{equation}

Brown's identity in \eqref{eq:brown_identity} allows us to translate the Fisher-information bounds derived
in the preceding section into upper bounds on the MMSE. We summarize the
resulting bounds below.

\begin{itemize}[leftmargin=*]
\item For a discrete random vector $\bfX$, Corollary~\ref{cor:renyi_family}
and Brown's identity give, for every $p\in[0,2]$,
\begin{equation}
{\bf MMSE}(\bfX\mid\bfY_\gamma)
\preceq
\frac{1-\beta_p(\bfX)}{\gamma}\bfI_n.
\label{eq:upper_MMSE_betap}
\end{equation}
Unlike the LMMSE bound, \eqref{eq:upper_MMSE_betap} depends on the
probability masses of $\bfX$ rather than on the locations of its support
points.

\item The dimension-aware bound in Theorem~\ref{thm:dimension_aware} gives
the following trace bound:
\begin{equation}
\mmse(\bfX\mid\bfY_\gamma)
\le \frac{1}{\gamma} \cdot
\begin{cases}
\displaystyle
1-a_1
\rme^{-2I_0(\bfX;\gamma)}
,
& n=1, \\
\displaystyle
n-a_n
\rme^{
-\frac{2}{n}
I_{\frac{1}{r_n}}(\bfX;r_n\gamma)
}
,
& n\ge2.
\end{cases}
\label{eq:alpha_mi_bound}
\end{equation}
If $\bfX$ is discrete, \eqref{eq:alpha_mi_bound} further
implies
\begin{equation}
\mmse(\bfX\mid\bfY_\gamma)
\le
\frac{1}{\gamma}
\left[
n-a_n
\exp\left(
-\frac{2}{n}H_{r_n}(\bfX)
\right)
\right]. \label{eq:MMSE_inequality_entropy}
\end{equation}

\item Finally, for \(n\ge4\), Theorem~\ref{thm:gaussian_kernel_fisher}
and Brown's identity give
\begin{equation}
\mmse(\bfX\mid\bfY_\gamma)
\le
\frac{1}{\gamma}
\left[
n-a_n
\exp\left(
-\frac{1}{2}
I_{\frac12}\left(
\bfX;\frac{8\gamma}{n}
\right)
\right)
\right].
\label{eq:mmse_order_half}
\end{equation}

\end{itemize}

Thus, the bounds above capture different properties of the input
distribution: \eqref{eq:lmmse_application} depends on covariance,
\eqref{eq:alpha_mi_bound} and \eqref{eq:mmse_order_half} depend on
the probability masses.

The bound can also be combined with the I-MMSE relationship to provide new bounds on the mutual information.  For example, for 
$n\ge2$ and $p\in[0,2]$ let 
\begin{equation}
\Delta_{n,p}(\bfX)
:=
\max\left\{
 n\beta_p(\bfX),
 a_n \rme^{-2H_{r_n}(\bfX)/n}
\right\}.
\label{eq:def_Delta_np}
\end{equation}
the by using the MMSE inequality in \eqref{eq:MMSE_inequality_entropy} together with I-MMSE relationship \eqref{eq:I_MMSE_prelim} one can show that 
\begin{align}
I(\bfX;\bfY_\gamma)
\le
\frac12\inf_{0<\gamma_0\le\gamma}\Bigg\{
&\log\det(\bfI_n+\gamma_0\bfK_{\bfX})
+\big(n-\Delta_{n,p}(\bfX)\big)
\log\frac{\gamma}{\gamma_0}
\Bigg\},
\label{eq:mutual_information_new}
\end{align}
where as before $\bfY_\gamma=\sqrt\gamma\bfX+\bfZ$ and $\bfK_{\bfX}$ is a covariance matrix of $\bfX$. The bound in \eqref{eq:mutual_information_new} is tighter than a classical $\frac12\log\det(\bfI_n+\gamma\bfK_{\bfX})$ for low and intermediate values of $\gamma$.

\section{Discussion and Conclusions}\label{sec:discussion}
This work develops several lower bounds on the Fisher information of Gaussian location mixtures. For discrete mixing distributions, the universal matrix inequality provides a lower bound in positive-semidefinite order in terms of the R\'enyi entropy of the mixing distribution. When the mixing distribution is supported on  $K$ points, the resulting bound scales as $\frac{\log K}{K^2}$ and therefore improves by a logarithmic factor on the support-size dependence obtained from Stam's inequality.

The dimension-aware inequalities provide complementary trace bounds that exploit the ambient dimension. These bounds are expressed in terms of Sibson's $\alpha$-mutual information and, for a $K$-point mixing distribution, scale as $K^{-2/n}$.

The inequalities are also specialized to rotationally invariant distributions supported on discrete spherical shells. An intrinsic-dimensional bound on Sibson mutual information separates the discrete radial complexity from the continuous angular component. When combined with the dimension-aware Fisher-information inequality, this gives a lower bound on the number of occupied shells. As an application, the number of shells in the least favorable prior for the vector bounded Gaussian mean problem is shown to grow at least linearly with the radius constraint. In the scalar case, the universal Fisher-information inequality yields an $\Omega(R\sqrt{\log R})$ lower bound on the number of mass points and an $O(1/R)$ upper bound on each probability mass. The Fisher-information inequalities also lead, through Brown's identity and the I-MMSE relationship, to new upper bounds on the MMSE and mutual information.

Several questions remain open. For example, the linear shell-count lower bound may not describe the exact growth rate of the number of shells in the vector bounded Gaussian mean problem. Determining this growth rate and obtaining sharper information about the locations and probabilities of the occupied shells are natural directions for future work.

\appendices

\section{Proof of \eqref{eq:sibson_half_gaussian_kernel} in Proposition~\ref{prop:sibson_properties}}
\label{app:proof_prop}

Setting $\alpha=\frac12$ in
\eqref{eq:sibson_renyi_entropy_representation} gives
\begin{equation}
I_{\frac12}(\bfX;\epsilon)
=
h_2\left(
\bfX+\sqrt{\frac{2}{\epsilon}}\bfZ
\right)
-
h_2\left(
\sqrt{\frac{2}{\epsilon}}\bfZ
\right).
\label{eq:sibson_half_entropy_difference}
\end{equation}
Let
\begin{equation}
\sigma^2=\frac{2}{\epsilon}
\end{equation}
and denote by $f_{\bfY}$ the density of
\begin{equation}
\bfY=\bfX+\sigma\bfZ.
\end{equation}
By the definition of R\'enyi entropy of order two,
\begin{equation}
h_2(\bfY)
=
-\log\int_{\bbR^n}f_{\bfY}(\bfy)^2\,\rmd\bfy.
\label{eq:h2_collision_density}
\end{equation}
Since
\begin{equation}
f_{\bfY}(\bfy)
=
\bbE\left[\phi_\sigma(\bfy-\bfX)\right],
\end{equation}
Fubini's theorem gives
\begin{align}
\int_{\bbR^n}f_{\bfY}(\bfy)^2\,\rmd\bfy
&=
\bbE\left[
\int_{\bbR^n}
\phi_\sigma(\bfy-\bfX)
\phi_\sigma(\bfy-\bfX')
\,\rmd\bfy
\right]
\\
&=
(4\pi\sigma^2)^{-n/2}
\bbE\left[
\exp\left(
-\frac{\|\bfX-\bfX'\|^2}{4\sigma^2}
\right)
\right],
\label{eq:gaussian_collision_identity}
\end{align}
where $\bfX'$ is an independent copy of $\bfX$. On the other hand,
\begin{equation}
\int_{\bbR^n}\phi_\sigma(\bfy)^2\,\rmd\bfy
=
(4\pi\sigma^2)^{-n/2}.
\label{eq:gaussian_collision_normalizer}
\end{equation}
Substituting \eqref{eq:gaussian_collision_identity} and
\eqref{eq:gaussian_collision_normalizer} into
\eqref{eq:sibson_half_entropy_difference}, and using
$\sigma^2=2/\epsilon$, yields
\begin{equation}
I_{\frac12}(\bfX;\epsilon)
=
-\log
\bbE\left[
\exp\left(
-\frac{\epsilon}{8}
\|\bfX-\bfX'\|^2
\right)
\right],
\end{equation}
which proves \eqref{eq:sibson_half_gaussian_kernel}.

\section{Projection Proof of Theorem~\ref{thm:universal_fisher_bound}}
\label{app:vector_projection}
It remains to justify the vector extension at $\epsilon=1$.  Define
\begin{equation}
b(H)= \rme^{-2H}\max\left\{\frac{2}{\pi},c_{\rm F}H\right\}.
\end{equation}
Write $\supp(P_{\bfX})=\{\bfx_j\}_{j\ge1}$.  Let $\bfu$ be a unit vector such
that
\begin{equation}
\bfu^T\bfx_i\ne\bfu^T\bfx_j,
\qquad i\ne j.
\end{equation}
Then $\bfu^T\bfX$ has exactly the same probability masses as $\bfX$, and
therefore
\begin{equation}
H_2(\bfu^T\bfX)=H_2(\bfX).
\end{equation}
Moreover, define $Z_{\bfu}:=\bfu^T\bfZ$.  Then
\begin{equation}
\bfu^T(\bfX+\bfZ)=\bfu^T\bfX+Z_{\bfu},
\qquad Z_{\bfu}\sim\cN(0,1).
\end{equation}
The scalar part of Theorem~\ref{thm:universal_fisher_bound} gives
\begin{equation}
J(\bfu^T(\bfX+\bfZ))\ge b(H_2(\bfX)).
\end{equation}
By the projection data-processing inequality in Lemma~\ref{lem:projection_fisher_DPI},
\begin{equation}
J(\bfu^T(\bfX+\bfZ))
\le
\bfu^T\bfJ(\bfX+\bfZ)\bfu.
\end{equation}
Hence
\begin{equation}
\bfu^T\bfJ(\bfX+\bfZ)\bfu
\ge b(H_2(\bfX))
\label{eq:generic_direction_app}
\end{equation}
for every injective direction.

For each pair $i\ne j$, define the set of exceptional directions
\begin{equation}
\mathcal E_{ij}
:=
\left\{
\bfu\in\mathbb S^{n-1}:
\bfu^T(\bfx_i-\bfx_j)=0
\right\}.
\end{equation}
Since $\bfx_i\ne\bfx_j$, the set $\mathcal E_{ij}$ is the intersection
of $\mathbb S^{n-1}$ with the hyperplane orthogonal to
$\bfx_i-\bfx_j$.  It is therefore a lower-dimensional great sphere and
has zero surface measure.  Because $\supp(P_{\bfX})$ is at most
countable, the full exceptional set
\begin{equation}
\mathcal E
:=
\bigcup_{i\ne j}\mathcal E_{ij}
\end{equation}
is a countable union of sets of zero surface measure and hence also has
zero surface measure.  Moreover, every nonempty open subset of
$\mathbb S^{n-1}$ has positive surface measure.  It follows that
$\mathbb S^{n-1}\setminus\mathcal E$ is dense in
$\mathbb S^{n-1}$.

Every $\bfu\in\mathbb S^{n-1}\setminus\mathcal E$ defines an injective
projection on $\supp(P_{\bfX})$, since
\begin{equation}
\bfu^T\bfx_i\ne\bfu^T\bfx_j,
\qquad i\ne j.
\end{equation}
Now fix an arbitrary unit vector $\bfu$.  By density, there exists a
sequence $\{\bfu_k\}_{k\ge1}\subseteq
\mathbb S^{n-1}\setminus\mathcal E$ such that $\bfu_k\to\bfu$.
Applying \eqref{eq:generic_direction_app} to each $\bfu_k$ gives
\begin{equation}
\bfu_k^T\bfJ(\bfX+\bfZ)\bfu_k
\ge b(H_2(\bfX)).
\end{equation}
Since the quadratic form
$\bfv\mapsto\bfv^T\bfJ(\bfX+\bfZ)\bfv$ is continuous, letting
$k\to\infty$ yields
\begin{equation}
\bfu^T\bfJ(\bfX+\bfZ)\bfu
\ge b(H_2(\bfX)).
\end{equation}
This holds for every unit vector $\bfu$, and is therefore equivalent to
the desired bound
\begin{equation}
\bfJ(\bfX+\bfZ)\succeq b(H_2(\bfX))\bfI_n.
\end{equation}

This concludes the proof of the extension to the vector case.

\section{Properties of the Constant $a_n$}
\label{app:properties_an}

Now suppose that $n\ge3$. Applying the sharp Sobolev inequality
\eqref{eq:sharp_sobolev_prelim} to
\begin{equation}
g=\sqrt{\phi_1},
\end{equation}
where $\phi_1$ is the density of $\cN(0,\bfI_n)$, and using
\begin{equation}
J(\bfZ)=4\|\nabla\sqrt{\phi_1}\|_2^2,
\end{equation}
gives
\begin{equation}
J(\bfZ)
\ge
n(n-2)|\mathbb S^n|^{2/n}
\left(
\int_{\bbR^n}\phi_1^{r_n}
\right)^{1/r_n}.
\label{eq:an_gaussian_sobolev}
\end{equation}
Since
\begin{equation}
J(\bfZ)=n
\end{equation}
and
\begin{equation}
\left(\int_{\bbR^n}\phi_1^{r_n}\right)^{1/r_n}
=
\frac{1}{2\pi}
\left(\frac{n-2}{n}\right)^{(n-2)/2},
\end{equation}
the right-hand side of \eqref{eq:an_gaussian_sobolev} is precisely
$a_n$.  Therefore,
\begin{equation}
a_n< n,
\end{equation}
where the strictness follows from the fact that a Gaussian is not an extremizer of the sharp Sobolev inequality for $n\ge 3$.

It remains to determine the large-dimensional behavior of $a_n$.
Using
\begin{equation}
|\mathbb S^n|
=
\frac{2\pi^{(n+1)/2}}
{\Gamma\left(\frac{n+1}{2}\right)},
\label{eq:sphere_area_an_appendix}
\end{equation}
we may write, for $n\ge3$,
\begin{equation}
\frac{a_n}{n}
=
\frac{n-2}{2\pi}
\left(
\frac{2\pi^{(n+1)/2}}
{\Gamma\left(\frac{n+1}{2}\right)}
\right)^{2/n}
\left(1-\frac{2}{n}\right)^{(n-2)/2}.
\label{eq:an_ratio_exact}
\end{equation}
Stirling's expansion
\begin{equation}
\log\Gamma(x)
=
\left(x-\frac12\right)\log x-x
+\frac12\log(2\pi)
+O\left(\frac1x\right)
\end{equation}
and the Taylor expansion
\begin{equation}
\log\left(1-\frac{2}{n}\right)
=
-\frac{2}{n}-\frac{2}{n^2}
+O\left(\frac{1}{n^3}\right)
\end{equation}
applied to \eqref{eq:an_ratio_exact} yield
\begin{equation}
\log\frac{a_n}{n}
=
-\frac{1-\log 2}{n}
+
O\left(\frac{\log n}{n^2}\right).
\label{eq:an_log_asymptotic}
\end{equation}
Exponentiating \eqref{eq:an_log_asymptotic} gives
\begin{equation}
\frac{a_n}{n}
=
1-\frac{1-\log 2}{n}
+
O\left(\frac{\log n}{n^2}\right).
\label{eq:an_ratio_asymptotic}
\end{equation}
Consequently,
\begin{equation}
\frac{a_n}{n} \to 1
\end{equation}
and
\begin{equation}
n-a_n \to  1-\log 2.
\end{equation}
Finally, \eqref{eq:an_log_asymptotic} also implies
\begin{equation}
\frac{n}{2}\log\frac{n}{a_n}
\to 
\frac{1-\log 2}{2},
\end{equation}
which completes the proof.

\section{Proof of Lemma~\ref{lem:sphere_partition}}
\label{app:sphere_partition}
Let $\bfu_1,\ldots,\bfu_N$ be a maximal $\delta/2$-separated subset of
$\mathbb S^{n-1}$. The spherical caps
\begin{equation}
\left\{
\bfu\in\mathbb S^{n-1}:
\|\bfu-\bfu_k\|<\frac{\delta}{4}
\right\},
\qquad k=1,\ldots,N,
\end{equation}
are pairwise disjoint. The geodesic radius of each cap is
\begin{equation}
\theta_\delta
=
2\arcsin\left(\frac{\delta}{8}\right)
\ge\frac{\delta}{4}.
\end{equation}
Since $\delta\le1$ and
$\sin\theta\ge2\theta/\pi$ for $\theta\in[0,\pi/2]$, the surface area
of each cap is bounded below by
\begin{align}
|\mathbb S^{n-2}|
\int_0^{\theta_\delta}\sin^{n-2}\theta\,\rmd\theta
&\ge
|\mathbb S^{n-2}|
\int_0^{\delta/4}
\left(\frac{2\theta}{\pi}\right)^{n-2}\rmd\theta\\
&=
\frac{|\mathbb S^{n-2}|}{n-1}
\left(\frac{2}{\pi}\right)^{n-2}
\left(\frac{\delta}{4}\right)^{n-1}.
\end{align}
The disjointness of the caps and
\eqref{eq:def_sphere_covering_constant} therefore imply
\begin{equation}
N\le M_n\delta^{-(n-1)}.
\end{equation}

By maximality, the Euclidean balls of radius $\delta/2$ centered at
$\bfu_1,\ldots,\bfu_N$ cover $\mathbb S^{n-1}$. Assign each point of
the sphere to a nearest center, breaking ties according to the smallest
index, and let $\mathcal C_k$ be the resulting Voronoi cells. These
cells form a measurable partition and satisfy
\begin{equation}
\|\bfu-\bfu_k\|\le\frac{\delta}{2},
\qquad \bfu\in\mathcal C_k.
\end{equation}
Consequently, for every $\bfu,\bfv\in\mathcal C_k$,
\begin{equation}
\|\bfu-\bfv\|
\le
\|\bfu-\bfu_k\|+\|\bfv-\bfu_k\|
\le\delta,
\end{equation}
which proves \eqref{eq:sphere_partition_diameter}.

\bibliographystyle{IEEEtran}
\bibliography{refs_reorganized}

@article{stam1959some,
  author  = {A. J. Stam},
  title   = {Some Inequalities Satisfied by the Quantities of Information of {F}isher and {S}hannon},
  journal = {Information and Control},
  volume  = {2},
  number  = {2},
  pages   = {101--112},
  year    = {1959},
  doi     = {10.1016/S0019-9958(59)90348-1}
}

@article{zhang2007inequalities,
  author  = {Zhengmin Zhang},
  title   = {Inequalities for Characteristic Functions Involving {F}isher Information},
  journal = {C. R. Math. Acad. Sci. Paris},
  volume  = {344},
  number  = {5},
  pages   = {327--330},
  year    = {2007},
  doi     = {10.1016/j.crma.2007.01.008}
}

@article{eskenazis2024entropy,
  author  = {Alexandros Eskenazis and Lampros Gavalakis},
  title   = {On the Entropy and Information of {G}aussian Mixtures},
  journal = {Mathematika},
  volume  = {70},
  number  = {2},
  pages   = {e12246},
  year    = {2024},
  doi     = {10.1112/mtk.12246}
}

@article{aubin1976problemes,
  author  = {Thierry Aubin},
  title   = {Probl\`emes Isop\'erim\'etriques et Espaces de {S}obolev},
  journal = {J. Differential Geom.},
  volume  = {11},
  number  = {4},
  pages   = {573--598},
  year    = {1976},
  doi     = {10.4310/JDG/1214433725}
}

@article{talenti1976best,
  author  = {Giorgio Talenti},
  title   = {Best Constant in {S}obolev Inequality},
  journal = {Ann. Mat. Pura Appl.},
  volume  = {110},
  number  = {1},
  pages   = {353--372},
  year    = {1976},
  doi     = {10.1007/BF02418013}
}

@article{yamano2012bounds,
  author  = {Takuya Yamano},
  title   = {Bounds for {F}isher Information and Its Production under Flow},
  journal = {J. Math. Phys.},
  volume  = {53},
  number  = {4},
  pages   = {043301},
  year    = {2012},
  doi     = {10.1063/1.3700757}
}

@article{beckner1975inequalities,
  author  = {William Beckner},
  title   = {Inequalities in {F}ourier Analysis},
  journal = {Ann. of Math.},
  volume  = {102},
  number  = {1},
  pages   = {159--182},
  year    = {1975}
}

@article{brascamp1976best,
  author  = {Herm Jan Brascamp and Elliott H. Lieb},
  title   = {Best Constants in {Y}oung's Inequality, Its Converse, and Its Generalization to More than Three Functions},
  journal = {Adv. Math.},
  volume  = {20},
  number  = {2},
  pages   = {151--173},
  year    = {1976},
  doi     = {10.1016/0001-8708(76)90184-5}
}

@article{dytso2023conditional,
  author  = {Alex Dytso and H. Vincent Poor and Shlomo Shamai (Shitz)},
  title   = {Conditional Mean Estimation in {G}aussian Noise: A Meta Derivative Identity with Applications},
  journal = {IEEE Trans. Inf. Theory},
  volume  = {69},
  number  = {3},
  pages   = {1883--1898},
  year    = {2023},
  doi     = {10.1109/TIT.2022.3216012}
}

@article{guo2005mutual,
  author  = {Dongning Guo and Shlomo Shamai (Shitz) and Sergio Verd\'u},
  title   = {Mutual Information and Minimum Mean-Square Error in {G}aussian Channels},
  journal = {IEEE Trans. Inf. Theory},
  volume  = {51},
  number  = {4},
  pages   = {1261--1282},
  year    = {2005},
  doi     = {10.1109/TIT.2005.844072}
}

@incollection{marchand2004review,
  author    = {\'Eric Marchand and William E. Strawderman},
  title     = {Estimation in Restricted Parameter Spaces: A Review},
  booktitle = {A Festschrift for Herman Rubin},
  series    = {IMS Lecture Notes--Monograph Series},
  volume    = {45},
  pages     = {21--44},
  publisher = {Institute of Mathematical Statistics},
  year      = {2004},
  doi       = {10.1214/lnms/1196285377}
}

@article{casella1981estimating,
  author  = {George Casella and William E. Strawderman},
  title   = {Estimating a Bounded Normal Mean},
  journal = {Ann. Statist.},
  volume  = {9},
  number  = {4},
  pages   = {870--878},
  year    = {1981},
  doi     = {10.1214/aos/1176345527}
}

@article{bickel1981minimax,
  author  = {Peter J. Bickel},
  title   = {Minimax Estimation of the Mean of a Normal Distribution when the Parameter Space Is Restricted},
  journal = {Ann. Statist.},
  volume  = {9},
  number  = {6},
  pages   = {1301--1309},
  year    = {1981},
  doi     = {10.1214/aos/1176345646}
}

@inproceedings{yagli2019estimation,
  author    = {Semih Yagli and Alex Dytso and H. Vincent Poor},
  title     = {Estimation of Bounded Normal Mean: An Alternative Proof for the Discreteness of the Least Favorable Prior},
  booktitle = {2019 IEEE Information Theory Workshop (ITW)},
  pages     = {1--5},
  year      = {2019},
  doi       = {10.1109/ITW44776.2019.8988927}
}

@inproceedings{dytso2018structure,
  author    = {Alex Dytso and Ronit Bustin and H. Vincent Poor and Shlomo Shamai (Shitz)},
  title     = {On the Structure of the Least Favorable Prior Distributions},
  booktitle = {2018 IEEE International Symposium on Information Theory (ISIT)},
  pages     = {1081--1085},
  year      = {2018},
  doi       = {10.1109/ISIT.2018.8437546}
}

@book{vantrees2001detection,
  author    = {Harry L. Van Trees},
  title     = {Detection, Estimation, and Modulation Theory, Part I},
  edition   = {2nd},
  publisher = {Wiley},
  address   = {New York},
  year      = {2001}
}

@article{perezcruz2010mimo,
  title={{MIMO} {G}aussian Channels With Arbitrary Inputs: Optimal Precoding and Power Allocation},
  author={P{\'e}rez-Cruz, Fernando and Rodrigues, Miguel R. D. and Verd{\'u}, Sergio},
  journal={IEEE Trans. Inf. Theory},
  volume={56},
  number={3},
  pages={1070--1084},
  year={2010},
  doi={10.1109/TIT.2009.2039045}
}

@article{alvarado2014high,
  title={High-{SNR} Asymptotics of Mutual Information for Discrete Constellations With Applications to {BICM}},
  author={Alvarado, Alex and Brannstrom, Fredrik and Agrell, Erik and Koch, Tobias},
  journal={IEEE Trans. Inf. Theory},
  volume={60},
  number={2},
  pages={1061--1076},
  year={2014},
  doi={10.1109/TIT.2013.2291865}
}

@article{alvarado2018asymptotic,
  title={Asymptotic Comparison of {ML} and {MAP} Detectors for Multidimensional Constellations},
  author={Alvarado, Alex and Agrell, Erik and Brannstrom, Fredrik},
  journal={IEEE Trans. Inf. Theory},
  volume={64},
  number={2},
  pages={1231--1240},
  year={2018},
  doi={10.1109/TIT.2017.2727521}
}

@article{berry1990minimax,
  title={Minimax estimation of a bounded normal mean vector},
  author={Berry, J Calvin},
  journal={Journal of Multivariate Analysis},
  volume={35},
  number={1},
  pages={130--139},
  year={1990},
  publisher={Elsevier}
}

@ARTICLE{vector_Gausss,
  author={Dytso, Alex and Al, Mert and Poor, H. Vincent and Shamai Shitz, Shlomo},
  journal={IEEE Transactions on Information Theory}, 
  title={On the Capacity of the Peak Power Constrained Vector {G}aussian Channel: {A}n Estimation Theoretic Perspective}, 
  year={2019},
  volume={65},
  number={6},
  pages={3907-3921},
  doi={10.1109/TIT.2018.2890208}}

@article{marchand2002minimax,
  title={On the minimax estimator of a bounded normal mean},
  author={Marchand, {\'E}ric and Perron, Fran{\c{c}}ois},
  journal={Statistics \& probability letters},
  volume={58},
  number={4},
  pages={327--333},
  year={2002},
  publisher={Elsevier}
}

@article{zamir1998proof,
  title={A proof of the {F}isher information inequality via a data processing argument},
  author={Zamir, Ram},
  journal={IEEE Trans. Inf. Theory},
  volume={44},
  number={3},
  pages={1246--1250},
  year={1998},
  publisher={IEEE}
}

@book{ladyzhenskaya1969mathematical,
	author    = {Ladyzhenskaya, Olga A.},
	title     = {The Mathematical Theory of Viscous Incompressible Flow},
	edition   = {2},
	publisher = {Gordon and Breach},
	address   = {New York},
	year      = {1969}
}

@article{bobkov2014fisher,
  title={Fisher information and the central limit theorem},
  author={Bobkov, Sergey G and Chistyakov, Gennadiy P and G{\"o}tze, Friedrich},
  journal={Probability theory and related fields},
  volume={159},
  number={1},
  pages={1--59},
  year={2014},
  publisher={Springer}
}

@incollection{bobkov2023entropic,
  title={Entropic isoperimetric inequalities},
  author={Bobkov, Sergey G and Roberto, Cyril},
  booktitle={High Dimensional Probability IX: The Ethereal Volume},
  pages={97--121},
  year={2023},
  publisher={Springer}
}

@inproceedings{verdu2015alpha,
  title={$\alpha$-mutual information},
  author={Verd{\'u}, Sergio},
  booktitle={2015 Information Theory and Applications Workshop (ITA)},
  pages={1--6},
  year={2015},
  organization={IEEE}
}

@article{esposito2025sibson,
  title={Sibson $\alpha$-mutual information and its variational representations},
  author={Esposito, Amedeo Roberto and Gastpar, Michael and Issa, Ibrahim},
  journal={IEEE Transactions on Information Theory},
  volume={72},
  number={7},
  pages={4434--4467},
  year={2025},
  publisher={IEEE}
}

@article{milanian2026alpha,
  title={$\alpha$-Mutual Information for the {G}aussian Noise Channel},
  author={Milanian, Mohammad and Dytso, Alex and Cardone, Martina},
  journal={arXiv preprint arXiv:2604.10922},
  year={2026}
}

@article{sibson1969information,
  title={Information radius},
  author={Sibson, Robin},
  journal={Zeitschrift f{\"u}r Wahrscheinlichkeitstheorie und verwandte Gebiete},
  volume={14},
  number={2},
  pages={149--160},
  year={1969},
  publisher={Springer}
}

@article{wu2011functional,
  title={Functional properties of minimum mean-square error and mutual information},
  author={Wu, Yihong and Verd{\'u}, Sergio},
  journal={IEEE Trans. Inf. Theory},
  volume={58},
  number={3},
  pages={1289--1301},
  year={2011},
  publisher={IEEE}
}

@article{bercher2009minimum,
  title={On minimum {F}isher information distributions with restricted support and fixed variance},
  author={Bercher, J-F and Vignat, Christophe},
  journal={Information Sciences},
  volume={179},
  number={22},
  pages={3832--3842},
  year={2009},
  publisher={Elsevier}
}

@ARTICLE{Wang_AMplitude,
  author={Wang, Haiyang and Barletta, Luca and Dytso, Alex},
  journal={IEEE Transactions on Information Theory}, 
  title={An Improved Lower Bound on Cardinality of Support of the Amplitude-Constrained AWGN Channel}, 
  year={2026},
  volume={72},
  number={8},
  pages={5392-5399},
  doi={10.1109/TIT.2026.3697668}}

@article{gigante2017diameter,
  author  = {Gigante, Giacomo and Leopardi, Paul},
  title   = {Diameter Bounded Equal Measure Partitions of {A}hlfors Regular Metric Measure Spaces},
  journal = {Discrete \& Computational Geometry},
  volume  = {57},
  number  = {2},
  pages   = {419--430},
  year    = {2017},
  doi     = {10.1007/s00454-016-9834-y}
}

@article{dehesa2012generalized,
  author  = {Dehesa, Jes{\'u}s S. and Plastino, Angel R. and
             S{\'a}nchez-Moreno, Pablo and Vignat, Christophe},
  title   = {Generalized {C}ram{\'e}r--{R}ao Relations for
             Non-Relativistic Quantum Systems},
  journal = {Applied Mathematics Letters},
  volume  = {25},
  number  = {11},
  pages   = {1689--1694},
  year    = {2012},
  doi     = {10.1016/j.aml.2012.01.038}
}

@inproceedings{dytso2024uniform,
  title={Uniform Distribution on ($ n-1$)-Sphere: {R}ate-Distortion Under Squared Error Distortion},
  author={Dytso, Alex and Cardone, Martina},
  booktitle={2024 IEEE International Symposium on Information Theory (ISIT)},
  pages={873--878},
  year={2024},
  organization={IEEE}
}

@article{marchand2004estimation,
  title={Estimation in restricted parameter spaces: {A} review},
  author={Marchand, Eric and Strawderman, William E},
  journal={Lecture notes-monograph series},
  pages={21--44},
  year={2004},
  publisher={JSTOR}
}

@article{ismail1995bounds,
  author  = {Ismail, Mourad E. H. and Muldoon, Martin E.},
  title   = {Bounds for the Small Real and Purely Imaginary Zeros of
             {B}essel and Related Functions},
  journal = {Methods and Applications of Analysis},
  volume  = {2},
  number  = {1},
  pages   = {1--21},
  year    = {1995},
  doi     = {10.4310/MAA.1995.v2.n1.a1}
}

@article{qu1999best,
  author  = {Qu, Changzheng and Wong, Roderick},
  title   = {{``Best Possible'' Upper and Lower Bounds for the Zeros
             of the {B}essel Function $J_\nu(x)$}},
  journal = {Transactions of the American Mathematical Society},
  volume  = {351},
  number  = {7},
  pages   = {2833--2859},
  year    = {1999},
  doi     = {10.1090/S0002-9947-99-02165-0}
}

@article{wu2025entropic,
  title={Entropic isoperimetric and {C}ram{\'e}r--{R}ao inequalities for {R}{\'e}nyi--{F}isher information},
  author={Wu, Hao and Yu, Lei},
  journal={IEEE Transactions on Information Theory},
  year={2025},
  publisher={IEEE}
}

@article{raginsky2013concentration,
  title={Concentration of measure inequalities in information theory, communications, and coding},
  author={Raginsky, Maxim and Sason, Igal},
  journal={Foundations and Trends in Communications and Information Theory},
  volume={10},
  number={1-2},
  pages={1--247},
  year={2013},
  publisher={Emerald Publishing Limited}
}

@article{dulek2017characteristic,
  title={Characteristic function-based lower bounds for {F}isher information under arbitrary parametrization},
  author={Dulek, Berkan},
  journal={IEEE Transactions on Aerospace and Electronic Systems},
  volume={53},
  number={1},
  pages={501--512},
  year={2017},
  publisher={IEEE}
}

@inproceedings{wu2010impact,
  title={The impact of constellation cardinality on {G}aussian channel capacity},
  author={Wu, Yihong and Verd{\'u}, Sergio},
  booktitle={2010 48th Annual Allerton Conference on Communication, Control, and Computing (Allerton)},
  pages={620--628},
  year={2010},
  organization={IEEE}
}

@article{dytso2016interference,
  title={Interference as noise: Friend or foe?},
  author={Dytso, Alex and Tuninetti, Daniela and Devroye, Natasha},
  journal={IEEE Transactions on Information Theory},
  volume={62},
  number={6},
  pages={3561--3596},
  year={2016},
  publisher={IEEE}
}

@ARTICLE{CodebookPaper,
  author={Dytso, Alex and Tuninetti, Daniela and Devroye, Natasha},
  journal={IEEE Transactions on Information Theory}, 
  title={On the Two-User Interference Channel With Lack of Knowledge of the Interference Codebook at One Receiver}, 
  year={2015},
  volume={61},
  number={3},
  pages={1257-1276},
  doi={10.1109/TIT.2015.2388481}}

@article{selim2015modeling,
  title={Modeling and analysis of wireless channels via the mixture of {G}aussian distribution},
  author={Selim, Bassant and Alhussein, Omar and Muhaidat, Sami and Karagiannidis, George K and Liang, Jie},
  journal={IEEE Transactions on Vehicular technology},
  volume={65},
  number={10},
  pages={8309--8321},
  year={2015},
  publisher={IEEE}
}

@ARTICLE{GaussianBoundsCard,
  author={Dytso, Alex and Yagli, Semih and Poor, H. Vincent and {Shamai (Shitz)}, Shlomo},
  journal=TransIT,  
  title={The Capacity Achieving Distribution for the Amplitude Constrained Additive {G}aussian Channel: {A}n Upper Bound on the Number of Mass Points}, 
  year={2020},
  volume={66},
  number={4},
  pages={2006-2022},
  doi={10.1109/TIT.2019.2948636}}

@article{melbourne2022differential,
  title={The differential entropy of mixtures: {N}ew bounds and applications},
  author={Melbourne, James and Talukdar, Saurav and Bhaban, Shreyas and Madiman, Mokshay and Salapaka, Murti V},
  journal={IEEE Transactions on Information Theory},
  volume={68},
  number={4},
  pages={2123--2146},
  year={2022},
  publisher={IEEE}
}

@ARTICLE{gMixturesKhandani,
  author={Moshksar, Kamyar and Khandani, Amir K.},
  journal={IEEE Transactions on Information Theory}, 
  title={Arbitrarily Tight Bounds on Differential Entropy of {G}aussian Mixtures}, 
  year={2016},
  volume={62},
  number={6},
  pages={3340-3354},
  doi={10.1109/TIT.2016.2553147}}

@ARTICLE{BestApporximationGaussianMixture,
  author={Ma, Yun and Wu, Yihong and Yang, Pengkun},
  journal=TransIT, 
  title={On the Best Approximation by Finite {G}aussian Mixtures}, 
  year={2025},
  volume={71},
  number={7},
  pages={5469-5492},
  doi={10.1109/TIT.2025.3558841}}

@article{tepedelenlioglu2010inequalities,
  title={On inequalities relating the characteristic function and {F}isher information},
  author={Tepedelenlioglu, Cihan and Banavar, Mahesh K and Spanias, Andreas},
  journal={arXiv preprint arXiv:1007.1483},
  year={2010}
}

@article{lutwak2012extensions,
  title={Extensions of {F}isher information and {S}tam's inequality},
  author={Lutwak, Erwin and Lv, Songjun and Yang, Deane and Zhang, Gaoyong},
  journal={IEEE Trans. Inf. Theory},
  volume={58},
  number={3},
  pages={1319--1327},
  year={2012},
  publisher={IEEE}
}

@article{olea2026approximate,
  title={Approximate Minimax Estimation of a Bounded Normal Mean via Stochastic Mirror Ascent},
  author={Olea, Jos{\'e} Luis Montiel and Zubova, Ekaterina},
  journal={arXiv preprint arXiv:2607.05350},
  year={2026}
}

\end{document}